\documentclass[12pt,centertags,reqno]{amsart}
\usepackage[foot]{amsaddr}
\usepackage[english]{babel}
\usepackage[T1]{fontenc}
\usepackage[numbers]{natbib}
\usepackage{amssymb}
\usepackage{fancyhdr}
\usepackage{url}
\usepackage{hyperref}
\usepackage{verbatim}
\usepackage{leftidx}
\usepackage{color,graphicx}

\usepackage{mathrsfs, mathtools}
\usepackage{stmaryrd}

\usepackage{marvosym}
\mathtoolsset{showonlyrefs}

\numberwithin{equation}{section} \makeatletter
\renewcommand{\subsection}{\@startsection
{subsection}{2}{0mm}{\baselineskip}{-0.25cm}
{\normalfont\normalsize\bf}} \makeatother

\graphicspath{{./figs/}}

\newtheorem{theorem}{Theorem}[section]

\newtheorem{proposition}[theorem]{Proposition}
\theoremstyle{definition}

\theoremstyle{remark}
\newtheorem{remark}[theorem]{Remark}

\def \P {\mathbf P}
\def \bF {\mathbb F}
\def \F {\mathcal F}
\def \R {\mathbb R}

\def \E {\mathbf E}

\usepackage{color}

\begin{document}

\title{The value of knowing the market price of risk
}


\author[K.~Colaneri]{Katia Colaneri}\address{Katia Colaneri, Department of Economics and Finance, University of Rome Tor Vergata, Via Columbia 2, 00133 Roma, IT.}\email{katia.colaneri@uniroma2.it}
\author[S.~Herzel]{Stefano Herzel}\address{Stefano Herzel, Department of Economics and Finance, University of Rome Tor Vergata, Via Columbia 2, 00133 Roma, IT.}\email{stefano.herzel@uniroma2.it}
\author[M.~Nicolosi]{Marco Nicolosi}\address{Marco Nicolosi, Department of Economics, University of Perugia, Via A. Pascoli 20, 06123 Perugia, IT.}\email{marco.nicolosi@unipg.it}

\date{}

\begin{abstract}
This paper presents an optimal allocation problem in a financial market with  one risk-free and  one risky asset, when the market is
driven by  a stochastic  market
price of risk.
We solve the problem in  continuous time, for an investor with a Constant Relative Risk Aversion (CRRA)  utility, under two scenarios: when the  market
price of risk is observable (the {\em  full information case}), and when it is not   (the {\em  partial information case}).
The corresponding market models are complete in the partial information case and incomplete in the other case, hence the two scenarios exhibit rather different features.
We  study how the access to
more accurate information on the market price of risk  affects the optimal strategies and
we determine the maximal price that the investor would be willing to pay to get  such information.
In particular, we  examine two  cases of additional information,
when an exact observation of the market price of risk is available either at time $0$ only (the {\em initial
information case}), or  during the whole investment period (the {\em dynamic information case}).

\end{abstract}

\maketitle


\section{Introduction}\label{s1}

Ours is a classical expected utility maximization problem in continuous time, first studied  by Merton (1969) \cite{merton1969}.
We solve it via the martingale method, proposed by Karatzas el al. (1987) \cite{karatzas1987} and by Cox and  Huang (1989)  \cite{Cox and Huang}. The martingale method relies on duality theory and  transforms the original dynamic  problem, usually solved through the Hamilton-Jacobi-Bellman (HJB) partial differential equation,
into an equivalent static optimization problem. It has two main advantages over the more direct HJB approach: it
leads to a quasi-linear partial differential equation that is usually easier to solve and
 it provides an expression of the optimal wealth as a function of the state price density, that can be used to relate
the optimal strategy to the current state of the market.

A necessary assumption to apply the standard martingale method is that the state price density is  unique, that is
the market is complete.
However for our problem, in the full information case,
the trading strategies  may  also depend on the  market price of risk, that is a not traded asset, and  this makes the corresponding market model
incomplete. Therefore we must rely on a modification of the standard approach, the so called minimax martingale method, introduced by He and Pearson (1991) \cite{HeandP}.
The minimax method exploits the fact that,  in an incomplete market, there are infinitely many state price densities but they all assign the same values to the marketable claims, i.e. those claims
 attainable by admissible trading strategies involving the market securities.
Hence, the optimal  final wealth is determined by selecting the state price density which minimizes the maximal expected utility  of the final
wealth, the {\em minimax state price density}.

We  model the stock  as a geometric brownian motion with a market price of risk given by an Ornstein-Uhlenbeck process, which is a Gaussian,  mean reverting process.
This is a convenient assumption, adopted by several studies which  will be mentioned below, and that may be justified by  empirical evidence.
To solve the investment problem under partial information it  is necessary to identify  the {\em filter}, that is the conditional distribution of the
unobservable process given the available information.
The assumptions on the model of the market
allows us
to apply the linear finite dimensional
Kalman filter
to identify the dynamics of the filter
 and characterize its conditional distribution, see, e.g. Lipster and Shiryaev (2001) \cite{lipster2001statistics}.
Then, following Fleming and Pardoux (1982) \cite{fleming1982optimal}, we transform the original optimization
problem into an equivalent one where all the state variables are adapted to the same filtration. Under this transformation the market model
is complete and the classical martingale method can be used, see e.g. Bj\"{o}rk, Davis and Landen (2010) \cite{BDL}.

Another important consequence of the assumptions on the  dynamics of the assets and on the utility of the investor is that the state variables of the market model, represented by the
(logarithm of) minimax state price density and the market price of risk in the full information case, and by the
unique state price density and the filter in the partial information case, are jointly affine.
This fact allows us to compute the corresponding optimal wealths (and strategies) in closed form, after solving  a system of Riccati equations that
is homogeneous under full information  and non-homogeneous under partial information.

We  apply the results for the full and partial information  problem to compute  the value of initial and  dynamic information. Of course,
by increasing the information set, the investor  gets a higher  expected optimal utility.
 To measure the  subjective value of such enlargements
 we compute the corresponding reservation prices.
 Again, because of the  structure of the model, their expressions are simple.
 The last part of the paper is devoted to numerical examples that illustrate a few applications of our results. Some  of the results are rather unexpected: to mention at least one of them, we will see that to maximize the Sharpe ratio of an investment, one should follow the strategy of a partially informed portfolio manager with CRRA utility rather than that of a fully informed one!

Our study relies on a long list of previous contributions, which we will mention below, but its closest references are   Kim and Omberg (1996)
\cite{kim}  and Brendle (2006)  \cite{brendle2006}, who study optimal investment problems similar to ours by using the  HJB approach,
showing that the HJB equations can be reduced to a system of Riccati equations.
  In particular, while
Kim and Omberg (1996) \cite{kim} are interested in the full information case for HARA utility functions, Brendle (2006)  \cite{brendle2006} also  focuses on  a partially informed investor endowed with bounded CRRA preferences, extending his analysis to a multi-dimensional market model.
We rely on many of their results, especially those related  to the solutions to the Riccati equations.
However, both  these papers
 fail to provide a verification theorem for their results, that is, they  only show  necessary, but not sufficient, conditions for the optimality of their proposed solutions.  In particular, Kim and Omberg (1996) \cite{kim} write
“{\em (...) we are not acquainted with any verification theorem that fits the model above, despite its relative simplicity (...)}”. They mention the fact that the classical
verification theorems cannot be applied because the
indirect utility that solves the Bellman equation is a function of the investor wealth that is
not restricted to a closed set. Hence, they can only solve numerically their HJB equation to suggest that, for a given choice of parameters, there should be no signs of multiple solutions.
Instead, by using the martingale approach, we prove  verification theorems for both the cases of our interest (Theorem \ref{VT} and \ref{VTpartial}), we apply them to our solutions to show that they are effectively optimal (Theorem \ref{prop:OptSolFullInfo} and \ref{prop:OptSolPartInfo}), and we provide sufficient conditions that are easy to verify directly on any set of  parameters (Proposition \ref{prop:bound} and \ref{prop:boundPI}) .

Having mentioned what we believe are the most important theoretical contributions of this paper, we summarize the other ones. We derive the distribution, conditional and unconditional, of the optimal wealth under full and under partial information at any time within the investment horizon (Proposition \ref{prop:MGF_FI} and \ref{prop:MGF_PI}). To the best of our knowledge such results are new and  have never been addressed before, despite the fact that the knowledge of the distribution of the optimal wealth may be useful for applications in portfolio and risk management.
Another novelty inspired the title of our paper, that is we assign a price to the information that an investor may buy from an expert who is able to provide the value of the market price of risk either at the beginning of the investment period or continuously in time. We hope that this result may provide a new tool to measure the level of uncertainty on the returns of an asset, to be used along with the standard ones based on volatility or implied volatility.
To support our theoretical findings we provide a rich numerical analysis that also shows interesting and sometimes unexpected results.
Last, but not least, by
solving the optimization problem under full information, we show a new application of the powerful
minimax martingale approach where it is possible to explicitly identify the minimax state price density and the associated penalization process. This is a  nice example that may be useful for  didactical purposes.

Before completing this introduction with a literature review, we
provide a short description of the rest of the paper.  In Section \ref{sec:market} we introduce the modeling framework. Section \ref{sec:fi} solves the optimization problem under full information, while
Section \ref{sec:partial_info} under partial
information. In
both sections we characterize the optimal investment strategy and provide a closed form representation for the optimal wealth as a function of the relevant state variables as well as
for its distribution. In Section \ref{sec:value} we define and compute the value of initial and dynamic information. A numerical study to illustrate the effects of the
parameters on the distribution of the wealth and on the value of information  is presented in Section \ref{sec:numerics}. Section \ref{sec:conclusions} concludes.

\subsection{Literature review}

Optimal investment problem in continuous time started from the work of Merton (1969) \cite{merton1969}, and extended since then in many directions with the scope of including more realistic situations. The extension
considered by us, when the drift of asset prices is not directly observable by the investor, leads to  problems of partial information. Problems of this type have been addressed considering several modelizations of the unobservable risk factors. Contributions in the case where prices are modeled as diffusions can be found for instance in Lackner (1995) \cite{lackner1995}, Lackner (1998) \cite{lackner1998} and Brendle (2006) \cite{brendle2006} under different approaches. In Brennan (1998) \cite{brennan1998} and Xia (2001) \cite{xia2001} the authors discuss the effect of learning on the portfolio choices. The setting where prices depend on an unobservable Markov chain is considered, for instance in B\"auerle and Rieder (2005) \cite{bauerle2005} and Haussmann and Sass (2004)
\cite{haussmann2004}  and  Barucci and Marazzina (2015) \cite{BarucciMarazzina2015}. 
Considering investors endowed with different levels of information motivates the work of Fouque, Papanicolaou, and Sircar  (2015) \cite{foque-1} who analyze the loss of utility due to partial information. In Frey, Gabih and Wunderlich (2012)
\cite{frey2012} and  Gabih et al. (2014) \cite{gabih} expert opinions in the form of signals at random discrete time points are investigated. This idea is extended to the case where expert opinions arrive continuously in time by
Fouque, Papanicolaou and Sircar (2017) \cite{foque} and by Davis and Lleo (2013)
\cite{davis} who implement the Black–Litterman model in a
continuous time setting and use separability of the filtering problem and the stochastic control problem to incorporate analyst views and non-tradeable assets as additional
source of observation to estimate the filter.
A similar setting  has been studied  by Danilova, Monoyios and Ng (2010) \cite{danilova2010} who assumed that the investor has partial information about the drift of the stock
price but also some privileged information about the future of stock price. In Putsch\"ogl and Sass (2008) \cite{sass2008} an optimal investment and consumption problem under partial observation is
analysed using Malliavin calculus.
Investment problems in a market with two cointegrated assets under partial information are studied in some recent works as for instance Lee and Papanicolaou (2016) \cite{lee2016} and Altay et al. (2018, 2019) \cite{altay2018}, \cite{altay2019}.

The issue of assessing the value of information is a classical one in economics and finance. Pikovsky and Karatzas (1996) \cite{PK} presented the problem of computing the
value of  initial information for an investor, defined as the informational gain in terms of incremental utility provided by the access to an enlarged filtration.
Amendinger, Becherer and Schweizer (2003) \cite{Amendinger2003} addressed the problem of computing the value of information for a trader who has the opportunity to buy some extra information. The problem is formulated for a complete market in the
mathematical framework of an initially enlarged filtration, and the value of information is derived via a comparison of the expected utility from terminal wealth.  
Chau, Cosso and Fontana (2018) \cite{Fontana2018} extended their approach to estimate the value of an insider information that may allow for an
arbitrage opportunity, assuming that unbounded profits cannot be reached with bounded risk.

We  model the market price of risk as an Ornstein-Uhlenbeck mean reverting process. We refer to Wachter (2002) \cite{wachter} for a review of the most important empirical contributions justifying such assumption. Wachter (2002) \cite{wachter}  solved the optimal investment and consumption problem in a complete model where the market price of risk and the stock return are perfectly negatively correlated.
A setting close to ours under full information where returns of the risky asset are driven by an Ornstein-Uhlenbeck process was proposed by Kim and Omberg (1996) \cite{kim}, who studied the portfolio optimization problem for a
HARA investor
and discussed the existence of the so called {\em nirvana solutions}, which happen when an infinite expected utility is
reached in finite time. An application of this setting to the problem of a fund manager whose compensation
depends on the relative performance with respect to a benchmark can be found in Nicolosi, Angelini and Herzel (2018) \cite{NHA}, and Herzel and Nicolosi (2019) \cite{HN2019}.

\section {The general setting} \label{sec:market}

Let $(\Omega, \mathcal{F}, \P)$ be a fixed probability space endowed with a complete and right continuous filtration ${\mathbb{F}
} = ({\mathcal{F}_t})_{\{0 \leq t \leq T \}}$ representing the global information flow where $T$ is a fixed time horizon. All processes defined below are assumed to be $\bF$-adapted. We consider a
market model with one risky asset $S$, \textit{the stock}, and one risk-free asset $B$ with
dynamics
\begin{eqnarray*}
\frac{d S_t}{S_t} & =  & \mu_t   dt + \sigma dZ^S_t \label{generalmodel1}  \quad S_0=s_0\in \R^+,\\
\frac{d B_t}{B_t} & =  & r   dt, \quad B_0=1
\label{generalmodel2}
\end{eqnarray*}
where  $\sigma>0$ and $r\geq 0$ are constant, $Z^S$ is a standard,  one dimensional, $\bF$-Brownian motion and
the drift process $\mu$  is of the form
 \begin{equation*}
\mu_t = r + \sigma X_t.
\label{drift}
\end{equation*}
The process $X$ represents the {\em market price of risk} $X$  and is assumed to follow an Ornstein-Uhlenbeck process with dynamics
\begin{equation*}
d{X_t} = -\lambda ({X_t}-\bar{X})dt+ \sigma_X dZ_t^X,
\label{XStoc}
\end{equation*}
where $X_0$ is a normally distributed random variable with mean $\pi_0$ and variance $R_0$,  $\lambda >0$ is a constant representing  the strength of attraction toward the
long term expected mean $\bar{X}\geq 0$, $\sigma_X > 0$ is  the volatility of
the market price of risk and  $Z^X$ is a  standard one-dimensional $\bF$-Brownian motion correlated with $Z^S$ with
\begin{equation*}
d\langle Z^X,Z^S \rangle_t = \rho\; dt,
\label{correlation}
\end{equation*}
for a constant correlation coefficient $\rho \in [-1,1]$. Let $Z^\bot$ be a Brownian motion independent of $Z^S$ such that $Z^X=\rho Z^S+ \sqrt{1-\rho^2} Z^\bot$. Then without loss of generality we can assume that $\bF$ is the complete and right continuous filtration generated by $ (Z^S, Z^\bot)$ .


An investor  trades the risky asset and the risk-free asset continuously in time, starting from an initial capital $w$, to maximize the  expected utility of her final wealth  at time $T$.
Her trading strategy  is self-financing and given by the process $\theta=\{\theta_t, \ t \in [0,T]\}$ representing the proportion of portfolio value  invested in the risky
asset at  time $t \in [0,T]$.
We assume the standard integrability condition on $\theta$
\begin{equation}\label{eq:integrability}
\E\left[\int_0^T\left(|\theta_t X_t|+\theta^2_t\right)   d t\right] <\infty.
\end{equation}
to ensure that the associated wealth process
\begin{equation}
\frac{dW_t}{W_t} = (r + \theta_t \sigma X_t)dt+\theta_t\sigma dZ^S_t, \quad W_0=w>0.
\label{BC}
\end{equation}
is well defined  (note that we also assume that no dividends are paid by the stock before time $T$) and to exclude arbitrage opportunities.
Further restrictions on the measurability of the process $\theta$ depending on the information set of the investor will be given in the next sections.

The investor has a power utility function
\begin{equation}
u(x) = \frac{1}{1-\gamma}x^{1-\gamma},  \label{utility}
\end{equation}
for every $x>0$ and for a positive risk aversion parameter $\gamma \neq 1$. By setting $\gamma=1$ we get the logarithmic utility $u(x)=\log x$. Note that for $\gamma>1$ the function $u(x)$ is bounded above while it becomes unbounded when $\gamma<1$.


We will solve optimization problems corresponding to
two different assumptions on the information flow available to the investor.
In the first case, we  assume that the investor observes both the stock price and  the market price of risk; therefore her information flow is given by the global filtration  $ {\mathbb{F}}=(\F_t)_{t \in [0,T]}$. In particular the initial information $\F_0$ is given by the sigma algebra generated by $X_0$ enlarged with the collection of $\P$-null sets.
In the second case, we assume that the investor directly observes stock prices but not the market price of risk. At any time $t \in [0,T]$, the value of $X$ has to be inferred from the available information, represented by  the natural filtration generated by the stock price process completed by the $\P$-null sets and
denoted by $\mathbb{F}^S$, where the initial information is $\F^S_0=\{\emptyset, \Omega\}$.

\section{Optimal investment under full information}\label{sec:fi}
In this section we consider a fully informed investor who observes the path of the stock   and  of the market price of risk. The investor wants to maximize the expected utility of her wealth at a time $T>0$, hence her problem is
\begin{equation}\label{pr:opt_full}
\max_{\theta \in \mathcal{A}(w)}\E\left[u(W_T)\right]
\end{equation}
where $\mathcal{A}(w)$ is the set of $\bF$-predictable self-financing strategies satisfying the  integrability condition \eqref{eq:integrability} starting from an initial wealth $w$.
Problem \eqref{pr:opt_full} is equivalent to
\begin{equation}\label{pr:opt_full2}
\max_{\theta \in \mathcal{A}(w)}\E\left[\frac{1}{1-\gamma}W_T^{1-\gamma}|\F_0\right] .
\end{equation}
In fact, a strategy $\theta^*\in \mathcal{A}(w)$ is optimal for \eqref{pr:opt_full} if and only if it is optimal for \eqref{pr:opt_full2} almost surely, because $\theta^*_0$ is $\F_0$-measurable. In the sequel we will denote by $\E_t$ the conditional expectation given $\mathcal F_t$.

In this setting there are two risk factors and only one  asset that can be used as a hedging instrument, therefore the market is not complete.
The
 state price densities $\xi^\nu$  satisfy the equation
\begin{equation}
\frac{d\xi_t^{\nu}}{\xi^\nu_t} = -r dt -X_t dZ_t^S - \nu_t \sqrt{1-\rho^2} dZ_t^\bot, \quad \xi_0^{\nu} = 1
\label{BC_eps}
\end{equation}
where $\nu=\{\nu_t, \ t \in [0,T]\}$ is such that $\E\left[\int_0^T \nu^2_t d t\right]<\infty$ and $\E\left[\int_0^T|\xi_t^{\nu}|^2dt\right]<\infty$. 

The process $Z_t^\bot$  is
orthogonal to the space of attainable payoffs (i.e. payoffs that can be reached by feasible self-financing strategies). If the stock price process and the market price of
risk are perfectly correlated (positively or negatively), $\sqrt{1-\rho^2}$ 
vanishes, the market becomes complete and the state price density is unique.

To solve the optimal investment problem in an incomplete market we
follow  He and
Pearson (1991) \cite{HeandP}
 and apply the  martingale approach to transform the dynamic problem \eqref{pr:opt_full2}  into the equivalent   static  one
\begin{equation}
\min_\nu \max_{W_T} \E_0[u(W_T)],
\label{static_pr}
\end{equation}
subject to the constraint
\begin{equation}
w = \E_0[\xi_T^\nu W_T] \label{static_BC}.
\end{equation}
The optimal  $\nu^*$ for the problem \eqref{static_pr}-\eqref{static_BC} determines the {\em minimax state price density} process $\xi^*$. The role of the  process
$\nu^*$  is to penalize contingent claims that cannot be replicated by feasible portfolio strategies. For example, $\nu^*_t=0$ $ \ \P$-a.s. for $t \in [0,T] $ implies that no penalization is necessary and a feasible optimal strategy is naturally chosen by the investor.
For power utility functions of the form \eqref{utility}, a sufficient condition for the existence of  $\xi^*$  is that $\gamma>1$, see He and Pearson (1991) \cite[Theorem 4
and Theorem 6]{HeandP}.

The Lagrangian function associated to problem \eqref{static_pr}-\eqref{static_BC} is
$${\mathcal L}(W_T,\lambda_0)= \E_0[u(W_T)] -\lambda_0 (\E_0[\xi_T^* W_T] - w),$$
where $\lambda_0$ is the  multiplier from the constraint \eqref{static_BC}. From standard results (e.g. Karatzas et al. (1987) \cite{karatzas1987}), the optimal final
wealth satisfies
\begin{equation}
W^*_T = g(\lambda_0 \xi^*_T ) \label{WT}
\end{equation}
where function  $g(\cdot)$ is the inverse of the marginal utility $u'$. For the power
utility  \eqref{utility} $g(y)=y^{-\frac{1}{\gamma}}$. Since ${\xi^*}$ is a state-price density process, the optimal wealth at
time $t$ is
\begin{eqnarray}
 W^*_t  &= &  {\xi^*_t}^{-1} \E_t \left[{\xi^*_T} g(\lambda_0 \xi^*_T )  \right] \nonumber \\
 &= & {\xi^*_t}^{-1} \lambda_0^{-\frac{1}{\gamma}} \E_t \left[{\xi^*_T}^{1-\frac{1}{\gamma}}\right] \label{eqWt}.
\end{eqnarray}

\begin{remark}
This approach can also be applied to logarithmic utility functions. In this case $g(y)=\frac{1}{y}$, the optimal wealth at
time $t$ is
\begin{eqnarray}
 W^*_t  =  {\xi^*_t}^{-1} \E_t \left[{\xi^*_T} g(\lambda_0 \xi^*_T )  \right]= (\lambda_0 \xi^*_t)^{-1}.\label{eq:Wlog}
\end{eqnarray}
By applying It\^o formula to \eqref{eq:Wlog} we get
\begin{align}
\frac{d W^*_t}{W^*_t}=&(r+ X^2_t+(1-\rho^2)(\nu^*_t)^2) d t + X_t dZ^S_t+\sqrt{1-\rho^2} \nu^*_t dZ_t^\bot.\label{eq:W*}
\end{align}
By equating the predictable quadratic covariations of $W^*$ and $Z^S$ computed from  \eqref{BC} and \eqref{eq:W*} we get the optimal strategy $\displaystyle \theta^*_t=\frac{X_t}{\sigma}$. This strategy is called myopic because it does not depend on the investment horizon.
Note that in this case the minimax state price density is associated to the penalty process $\nu^*_t=0$.
\end{remark}

The following verification theorem states sufficient conditions to solve the  optimization problem. 

\begin{theorem}[Verification Theorem under full information] \label{VT}
Let the function $F(Y, X, t)$ be the solution to the partial differential equation (where  subscripts denote partial derivatives)
\begin{eqnarray}
&\frac{1}{2}&F_{YY}Y^2 \left( X^2 + {\nu^*(Y,X,t)}^2 (1-\rho^2) \right) +  F_{XY}Y \sigma_X\left(\rho X + {\nu^*(Y,X,t)} (1-\rho^2)\right)
\nonumber \\
&+ &   \frac{1}{2}F_{XX}\sigma_X^2  - F_X \left(\sigma_X \left(\rho X + {\nu^*(Y,X,t)} (1-\rho^2)\right) + \lambda(X-\bar X) \right) \nonumber \\
&+ &   F_t   - rF + rF_Y Y  = 0
\label{eqpde}
\end{eqnarray}
where
\begin{eqnarray}
\nu^*(Y,X,t) = - \frac{\sigma_X F_X(Y,X,t)}{YF_Y(Y,X,t)} \label{c3}
\end{eqnarray}
with the boundary conditions
\begin{eqnarray}
F(Y,X,T) = Y^{\frac{1}{\gamma}} \label{c1}\\
F(Y_0, X_0,0)= w \label{c2}
\end{eqnarray}
for some constant $Y_0 > 0$, $F_Y(Y,X,t)\neq 0$, and $F(Y,X,t)\rightarrow F(Y,X,T)$ as $t\rightarrow T$.

Assume that the following conditions hold:
\begin{itemize}
\item[(i)] the function $\nu^*(Y,X,t)$
is sublinear and locally Lipschitz for $(Y,X,t)\in \R^+ \times \R \times [0,T]$
(this condition implies that the process $\xi^*_t$ satisfying \eqref{BC_eps} with $\nu^*_t:=\nu^*(Y_0(\xi^*_t)^{-1},X_t,t)$ is a well defined local martingale),
\item[(ii)]
\begin{equation}
\E\left[\int_0^T\!\!\!\!\left((Y_t^{-1}F(Y_t, X_t, t))^2 + (F_Y(Y_t, X_t, t))^2\right) (X_t^2+(\nu^*(Y_t, X_t, t))^2) d t\right]\!<\!\infty\label{eq:integral},
\end{equation}
where
$Y_t:=Y_0(\xi^*_t)^{-1}$, for $t \in [0,T]$.
\end{itemize}

Then the process $\xi^*_t$  is the minimax state price density, the optimal wealth is
$W^*_t=F(Y_t,X_t,t)$ and the optimal investment strategy is
\begin{align}
\theta_t^*=\frac{F_{Y}(Y_t, X_t, t)Y_t X_t+ \rho \sigma_X F_{X}(Y_t, X_t, t)}{\sigma F(Y_t, X_t, t)},\label{eq:optimal_strategy}
\end{align}
for every $t \in [0,T]$.
\end{theorem}

\noindent
\begin{proof}
To show that $F(Y_t,X_t,t)$ is the optimal wealth process, we need to verify that the initial wealth satisfies the budget constraint  and that the
final wealth satisfies the first order condition \eqref{WT} and is attainable by  a self-financing strategy.

The budget constraint \eqref{static_BC} follows from \eqref{c2}, and the first  order condition from \eqref{c1}, since $g(y)=y^{-\frac{1}{\gamma}}$. Since $\xi^*_t$ is well defined by condition (i) and $Y_0$ is given by \eqref{c2}, we can define the process $Y_t:=Y_0 \left(\xi^*_t\right)^{-1}$.  From It\^o formula the dynamics of $Y_t$ is
\begin{align*}
\frac{d Y_t}{Y_t}=&(r+ X^2_t+(1-\rho^2)(\nu^*(Y_t, X_t, t))^2) d t + X_t dZ^S_t+\sqrt{1-\rho^2} \nu^*(Y_t, X_t, t) dZ_t^\bot,
\end{align*}
and, still applying It\^o,
\begin{align}
&F(Y_t, X_t, t)=F(Y_0, X_0, 0)+ \int_0^t\mathcal{L}F(Y_s, X_s, s) d s \nonumber \\
&\quad +\int_0^t\left(\rho \sigma_X F_X(Y_s, X_s, s)+Y_s X_s F_Y(Y_s, X_s, s)\right) d Z^S_s  \nonumber \\
&\quad  + \int_0^t\sqrt{1-\rho^2}\left(\sigma_X F_X(Y_s, X_s, s)+Y_s \nu^*(Y_s, X_s, s) F_Y(Y_s, X_s, s)\right) d Z^{\bot}_s,  \label{eqF}
\end{align}
where $\mathcal L$ is the differential operator
\begin{align*}
\mathcal L F=&F_t +\frac{1}{2}F_{YY}Y^2 \left( X^2 + {\nu^*(Y,X,t)}^2 (1-\rho^2) \right)\nonumber\\
& + \sigma_X F_{XY}Y \left(\rho X + {\nu^*(Y,X,t)} (1-\rho^2)\right)
\nonumber \\
&+ \frac{1}{2}F_{XX}\sigma_X^2  - F_X  \lambda(X-\bar X)+  F_Y Y (r+X^2+ (\nu^*(Y,X,t))^2(1-\rho^2)).\nonumber
\end{align*}
By \eqref{c3}, the integral with respect to $Z^\bot$ in \eqref{eqF}  vanishes, therefore the final wealth $F(Y_T, X_T, T)$ belongs to the space of attainable payoffs.
To show that it can be obtained  by a self-financing strategy starting from $w$
it remains to show that  the process $\xi^*_tF(Y_t,X_t, t)$ is a true martingale.
From  It\^o formula and  assumption \eqref{eqpde}, we get
\begin{align*}
&\xi^*_t F(Y_t, X_t, t)=F(Y_0, X_0, 0)\\
&+ \int_0^t\xi^*_s\left(\rho \sigma_X F_X(Y_s, X_s, s) +Y_s X_s F_Y(Y_s, X_s, s)-F(Y_s, X_s, s) X_s\right) d Z^S_s\\
& -\sqrt{1-\rho^2} \int_0^t\xi^*_s F(Y_s, X_s, s) \nu^*(Y_s, X_s, s) d Z^{\bot}_s,
\end{align*}
which is a true martingale because
\begin{align*}
&\E\left[\int_0^T\left((\xi^*_s)^2\left(\rho \sigma_X F_X(Y_s, X_s, s) +Y_s X_s F_Y(Y_s, X_s, s)-F(Y_s, X_s, s) X_s\right)^2\right.\right. \\
&\qquad \left.\left.+(1-\rho^2
)(\xi^*_s)^2 (F(Y_s, X_s, s))^2 (\nu^*(Y_s, X_s, s))^2 \right)dt\right]\\
&=\E\left[\int_0^T\!\!\!\left((\xi^*_s)^2\left( Y_s F_Y(Y_s, X_s, s) (X_s-\rho\nu^*(Y_s, X_s, s))-F(Y_s, X_s, s) X_s\right)^2\right.\right.\\
&\qquad \left.\left.+(1-\rho^2)(\xi^*_s)^2 (F(Y_s, X_s, s))^2 (\nu^*(Y_s, X_s, s))^2\right) dt\right]
\end{align*}
\begin{align*}
&\leq c_1\E\left[\int_0^T \left((X^2_s + (\nu^*(Y_s, X_s, s))^2) (F_Y(Y_s, X_s, s))^2 \right.\right.\\
&\qquad\left.\left.+(\xi_s^*)^2 (F(Y_s, X_s, s))^2(X^2_s + (\nu^*(Y_s, X_s, s))^2)\right)dt\right],
\end{align*}
that is bounded by \eqref{eq:integral} ($c_1$ is a positive constant). Note that
 the first equality comes from \eqref{c3}, and in the inequality  we have used $(a+b)^2\leq 2 (a^2+ b^2)$, $\rho^2<1$, $1-\rho^2<1$, and the definition of $Y_t$.

Therefore,  $W^*_t=F(Y_t, X_t, t)$ is the optimal wealth process and $\xi^*_t$ is the minimax state price density (see He and Pearson (1991) \cite[Theorem 8]{HeandP}). Finally, by equating the predictable quadratic covariations of $W^*$ and $Z^S$ computed from  \eqref{BC} and \eqref{eqF} we get the optimal strategy \eqref{eq:optimal_strategy}.
\end{proof}

To determine a closed form expression for $W^*_t$ we guess that the joint process  $(\log(\xi^*), X, X^2)$ is affine. From this guess it follows that the conditional expectation in  \eqref{eqWt} is
 \begin{equation*}\label{eq:expectation}
\E_t\left[{\xi_T^*}^{1-\frac{1}{\gamma}}\right] = {\xi^*_t}^{1-\frac{1}{\gamma}} e^{A(t)+B(t)X_t+\frac{1}{2}C(t)X_t^2}
\end{equation*}
where the functions $A(t)$, $B(t)$ and $C(t)$ satisfy the system of Riccati equations
\begin{equation}\label{eqa}
\left\{
\begin{split}
\frac{dC}{dt} &= -a - b C(t) - c C(t)^2,\\
\frac{dB}{dt} &=  -C(t)\lambda\bar{X} - \left(\frac{b}{2}+c C(t)\right)B(t),\\
\frac{dA}{dt} &= \frac{\gamma-1}{\gamma}r - B(t)\lambda\bar{X} - \frac{1}{2}C(t)\sigma_X^2 - \frac{1}{2}c B(t)^2
\end{split}\right.
\end{equation}
with boundary conditions
\begin{equation}
A(T) = B(T) = C(T) = 0 ,
\label{BC_Riccati}
\end{equation}
for constants
\begin{align*}
a = \frac{1-\gamma}{\gamma^2}, \quad b =2\left(-\lambda +\frac{1-\gamma}{\gamma}\rho\sigma_X\right), \quad
c = \sigma_X^2\left(\rho^2+\gamma(1-\rho^2)\right).\label{abc}
\end{align*}

To prove that our guess is correct we must show that it  satisfies Theorem \ref{VT}.
Before that, we discuss the behavior of the solutions to the problem \eqref{eqa}-\eqref{BC_Riccati} without reporting them explicitly, as they can be found, for instance, in  Kim and Omberg (1996) \cite{kim}.

Let us define $\displaystyle \Delta := b^2-4ac = 4\left(p-\frac{q}{\gamma}\right),$
where
$\displaystyle p := \lambda^2 + 2\lambda\rho\sigma_X+\sigma_X^2,$ and $\displaystyle q:= 2\lambda\rho\sigma_X+\sigma_X^2\label{q}.$
It is easily seen that $p > q$ and $p \geq 0$. In particular, if $\rho \neq -1$ and $\sigma_X\neq \lambda$, then $p \neq 0$.
We define the critical correlation value
\begin{equation}
\rho^* = -\frac{\sigma_X}{2\lambda}\vee -1,
\label{rho_crit}
\end{equation}
and, for $ \rho \geq\rho^*$, the critical risk aversion parameter
\begin{equation}
\gamma^* = \frac{q}{p}.
\label{gamma_crit}
\end{equation}
Note that $0 \leq \gamma^* < 1$, where the lower bound follows from the assumption on $\rho$.
According to the classification by Kim and Omberg (1996) \cite{kim}, there are four possible cases: 
\begin{itemize}
\item[i.] If $\rho \geq \rho^*$ and $\gamma^* < \gamma < 1 \cup \gamma > 1$, then
  $\Delta >0$  and
the solution, called ``well-behaved normal", exists for every $t$ in [0,T].
\item[ii.] If $\rho > \rho^*$ and $0 < \gamma < \gamma^*$, then $\Delta < 0$ . The  solution is called  ``tangent" and is defined on $[0,T^*)$, where
\begin{equation}
T^* = \frac{\pi}{\eta} - \frac{2}{\eta}\arctan{\frac{b}{\eta}}
\label{crit_t}
\end{equation}
with $\eta = \sqrt{-\Delta}$.
In such a case the investment horizon $T$ has to be lower than $T^*$ in order that solution exists over the entire interval $[0,T]$.
\item[iii.] If $\rho > \rho^*$ and $\gamma = \gamma^*$,  we get the ``well-behaved hyperbolic" solution which exists
for every $t$ in [0,T]. This case corresponds to $\Delta = 0$ and $b < 0$.
\item[iv.] If $\rho < \rho^*$ and $\gamma > 0, \gamma \neq 1$, then $q<0$ and hence $\Delta >0$.
The solution is  ``well-behaved normal" and exists
for every $t$ in $[0,T]$.
\end{itemize}
Wachter (2002) \cite{wachter} noted that, for real market data, the correlation between  returns and the market price of risk is
usually negative and close to -1, hence case (iv) should be the most relevant for financial applications.

After providing the conditions under which the system of Riccati equations has a solution, we can state some sufficient conditions  for  our guess to provide the optimal wealth and the optimal policy for the full information case.
\begin{theorem}\label{prop:OptSolFullInfo}
Let the functions $A(t)$, $B(t)$ and $C(t)$ satisfy \eqref{eqa}- \eqref{BC_Riccati} on $[0,T]$, let
\begin{eqnarray}
\nu^*_t & := & -\gamma  \left(B(t)+C(t)X_t\right)\sigma_X \label{nu_1},\\
\lambda_0 & := & \left[ \frac{e^{A(0)+B(0)X_0+\frac{1}{2}C(0)X_0^2}}{w}\right]^\gamma ,
\label{K}
\end{eqnarray}
and let $\xi^*_t$ be the state price density process associated to $\nu^*_t$.

If
\begin{equation}
\E\left[\int_0^T \left((\xi^*_t)^{1-\frac{1}{\gamma}}e^{A(t)+B(t)X_t+\frac{1}{2}C(t)X_t^2}\right)^2 (1+X_t^2) dt\right]<\infty,
\label{eq:integral2}
\end{equation}
then $\xi^*_t$ is the minimax state price density, $\lambda_0$ is the Lagrange multiplier for  problem \eqref{static_pr}-\eqref{static_BC}, and
\begin{equation}
W_t^* = \left(\lambda_0 \xi^*_t\right)^{-\frac{1}{\gamma}} e^{A(t)+B(t)X_t+\frac{1}{2}C(t)X_t^2},
\label{OptW_t_a}
\end{equation}
\begin{equation*}
\theta_t^* = \frac{1}{\gamma}\frac{X_t}{\sigma} +\rho\frac{\sigma_X}{\sigma}(B(t)+C(t)X_t)
\label{OptStr1}
\end{equation*}
are the optimal wealth and the optimal strategy.
\end{theorem}

\begin{proof}
Let us define
\begin{equation}\label{eq:F}
F(Y,X,t):=  {Y}^{\frac{1}{\gamma}} e^{A(t)+B(t)X+\frac{1}{2}C(t)X^2},
\end{equation}
for  $(Y,X,t)\in \R^+\times \R\times [0,T]$. We need to check that the function $F$ satisfies the assumptions of Theorem \ref{VT}, and hence the optimal wealth process is equal to $F(Y_t,X_t,t)$, where $Y_t := Y_0 (\xi_t^*)^{-1}$.

From assumption \eqref{nu_1} we see that $\nu^*$  verifies Equation \eqref{c3}; moreover, it is sublinear and locally Lipschitz for $(Y,X,t)\in \R^+ \times \R \times [0,T]$, hence  condition (i) of Theorem \ref{VT} is satisfied.
Since the functions $A(t)$, $B(t)$ and $C(t)$ satisfy \eqref{eqa},
function $F(Y,X,t)$ solves \eqref{eqpde}. Moreover $F_Y(Y,X,t)\neq 0$.
The boundary conditions \eqref{BC_Riccati} imply that
$$F(Y,X,T)={Y}^{\frac{1}{\gamma}},$$
and therefore condition \eqref{c1} is also true.
Moreover, imposing  the budget constraint
$${Y_0}^{\frac{1}{\gamma}} e^{A(0)+B(0)X_0+\frac{1}{2}C(0)X_0^2}=w$$
we get that the value of $Y_0 >0$ that satisfies condition \eqref{c2} is $Y_0 = {\lambda_0}^{-1}$, where
$\lambda_0$ is given by \eqref{K}.
From \eqref{WT} it follows that  $\lambda_0$  is the Lagrange multiplier for  problem \eqref{static_pr}-\eqref{static_BC}.

To check that condition \eqref{eq:integral} in Theorem \ref{VT} is satisfied  we use
the fact that  $Y_t=(\lambda_0\xi_t^*)^{-1}$, and the definitions of $F$ given in \eqref{eq:F} and of $\nu^*$ in \eqref{nu_1},  to get
\begin{align*}
&\E\left[\int_0^T\!\!\!\!\left((Y_t^{-1}F(Y_t, X_t, t))^2 + (F_Y(Y_t, X_t, t))^2\right) (X_t^2+(\nu^*(Y_t, X_t, t))^2) d t\right]\\
&=\E\left[\int_0^T\!\!\!\!\left((\lambda_0 \xi_t^*)^{1-1/\gamma}  e^{A(t)+B(t)X+\frac{1}{2}C(t)X^2})^2 \right) (X_t^2+\gamma^2 \sigma_X^2(B(t)+C(t)X_t)^2) d t\right]\\
&\leq c_2 \E\left[\int_0^T\!\!\!\!\left((\xi_t^*)^{1-1/\gamma}  e^{A(t)+B(t)X+\frac{1}{2}C(t)X^2})^2 \right) (1+X_t^2) d t\right]
\end{align*}
for some constant $c_2>0$, because $B(t)$ and $C(t)$ are continuous functions on $[0,T]$. The last term
 is bounded by assumption \eqref{eq:integral2}.  This completes the proof.
\end{proof}

The following result provides some conditions for Theorem \ref{prop:OptSolFullInfo} that are easier to check than \eqref{eq:integral2}.
\begin{proposition}\label{prop:bound} If at least one of the following two holds:
\begin{enumerate}
\item[(i)]
$\gamma>1$
\item[(ii)] the functions $A(t)$, $B(t)$ and $C(t)$ satisfy \eqref{eqa}- \eqref{BC_Riccati} on $[0,T]$ and
\begin{equation}\label{eq:cond0}
1-4C(0)\max\left(R_0, R_0e^{-2\lambda T}+\frac{\sigma_X^2}{2\lambda}(1-e^{-2\lambda T})\right)>0.
\end{equation}
\end{enumerate}
Then all assumptions of  Theorem \ref{prop:OptSolFullInfo} are verified.
\end{proposition}

\begin{proof}
Recall that for $\gamma>1$ the functions $A,B,C$ are well defined on $[0,T]$.

Then, we only need to show that condition \eqref{eq:integral2} is satisfied.
By Cauchy Schwartz inequality, using Fubini and $Y_t=(\lambda_0 \xi^*_t)^{-1}$,
\begin{align*}
&\E\left[\int_0^T e^{2A(t)} \left({\xi^*_t}^{1-\frac{1}{\gamma}}\right)^2 e^{2B(t)X_t+C(t)X_t^2} (1+X_t^2) dt  \right]\\
&\quad \leq \kappa \int_0^T e^{2A(t)} \E\left[{\xi^*_t}^{8(1-\frac{1}{\gamma})}\right]^{\frac{1}{4}} \E\left[e^{4B(t)X_t+2C(t)X_t^2}\right]^{\frac{1}{2}} \E[(1+X_t^2)^4]^{\frac{1}{4}} dt,
\end{align*}
where $k$ is a positive constant. Considering each expectation separately, first we have
\begin{equation*}\label{eq:terza}
\E\left[{\xi^*_t}^{8(1-\frac{1}{\gamma})}\right]<\infty,\end{equation*}
for every $t \in [0,T]$, since $X$ is an Ornstein-Uhlenbeck process (see, e.g. Revuz and Yor (2013) \cite[Chapter 8, Ex. 3.14]{revuz}).
Second,
\[
\E[(1+X_t^2)^4]<\infty
\]
for every $t\in [0,T]$, since $X_t$ is a Gaussian random variable and hence has moments of all orders.
Finally
$$\E\left[e^{4B(t)X_t+2C(t)X_t^2}\right]<\infty$$
for every $t \in [0,T]$ if and only if $1-4C(t)v_t>0$, where $v_t=R_0e^{-2\lambda t}+\frac{\sigma_X^2}{2\lambda}(1-e^{-2\lambda t})$ is the variance of $X_t$. \footnote{\label{footnoteE} For a random variable  $\varepsilon\sim\mathcal{N}(\mu,\sigma^2)$,  if $1-c\sigma^2 >0$,
$\E[e^{a+b\varepsilon+\frac{1}{2}c \varepsilon^2}] = \frac{1}{\sqrt{1-c\sigma^2}} \exp{\left(a +\frac{b^2\sigma^2}{2(1-c \sigma^2)} + \frac{b \mu}{1-c\sigma^2}+
\frac{c\mu^2}{2(1-c
\sigma^2)}\right)}$.}

To show that $1-4C(t)v_t>0$ we use that $C(t)$ is strictly negative and increasing on [0,T] if $\gamma > 1$, and is strictly positive and decreasing if $\gamma <1$ (see Kim and Omberg (1996) \cite[Equation (23)]{kim}).
Then, for $\gamma>1$, $C(t)<0$, therefore $1-4C(t)v_t>0$. When $\gamma<1$,  $C(t)$ is positive and decreasing, hence $C(t)<C(0)$. Moreover, let $v_\infty:=\frac{\sigma^2_X}{2\lambda}$, then $v_t$ is increasing and $R_0 \leq v_t\leq v_T$ if $R_0<v_\infty$ and decreasing  with $v_T \leq v_t\leq R_0$ otherwise. This means that  $1-4C(t)v_t>1-4C(0)\max\left(R_0, v_T\right)$. The result then follows from \eqref{eq:cond0}.
\end{proof}

Condition \eqref{eq:cond0} can be easily verified on any set of parameters, but it is more
restrictive than Condition \eqref{eq:integral2}, that is more difficult to check. In the section devoted to the applications we
show graphically, in
Figure \ref{Criticaltime}, how much restrictive Condition \eqref{eq:cond0} is with respect to the  domain of existence of the corresponding system of Riccati equations.

Kim and Omberg (1996) \cite{kim}  and  Brendle (2006) \cite{brendle2006}  solved a problem similar to ours by using the Hamilton-Jacobi-Bellman (HJB) approach.
To recover their results, we compute the expected optimal utility at time $t$
\begin{eqnarray}
\E_t \left[ u(W_T^*)\right] &= & \frac{1}{1-\gamma} \lambda_0^{1-\frac{1}{\gamma}}
 \E_t\left[(\xi_T^*)^{1-\frac{1}{\gamma}}\right]  \label{eq:EU_t_1} \\
 & = &  \frac{1}{1-\gamma}  \lambda_0 \xi_t^* W_t^*  \label{eq:EU_t_2} \\
 & = &  \frac{1}{1-\gamma}  {W^*_t}^{1-{\gamma}} e^{\gamma(A(t)+B(t)X_t+\frac{1}{2}C(t)X_t^2)}, \label{eq:EU_t}
  \end{eqnarray}
  where \eqref{eq:EU_t_1} follows from \eqref{WT}, \eqref{eq:EU_t_2} from \eqref{eqWt}, and \eqref{eq:EU_t} from \eqref{OptW_t_a}.
Equation  \eqref{eq:EU_t} corresponds to the  formulas \cite[Equation (16)]{kim} and  \cite[Equation (14)]{brendle2006}.

By plugging \eqref{OptW_t_a} into \eqref{eq:EU_t}, we can also compute the conditional expected optimal utility as
\begin{equation}
\E_t\left[u(W_T^*)\right] =  \frac{1}{1-\gamma}{\left(\lambda_0 \xi^*_t\right)} ^{1-1/\gamma}
e^{A(t)+B(t)X_t+\frac{1}{2}C(t)X_t^2}.
\label{phisT}\end{equation}
An advantage of  the martingale approach  over HJB  is that it  allows us  to compute both the optimal wealth  \eqref{OptW_t_a} and the expected optimal utility \eqref{phisT} as  functions of the minimax state price density $\xi^*$ and of the market price of risk $X$. This may be useful to study the dependence on the current state of the market.

From \eqref{eq:EU_t}, we can also derive the (unconditional) expected optimal utility, that exists when
$$Q(0):=1-\gamma C(0)R_0
$$
is strictly positive under the hypotheses of Proposition \ref{prop:bound}. \footnote{In the next section we will provide further conditions for the positiveness of $Q(0)$ (see Proposition \ref{prop:Q}).}
Indeed, since $X_0\sim N(\pi_0, R_0)$, and using the formula provided in Footnote \ref{footnoteE}, we get
\begin{eqnarray}
\E \left[u(W_T^*)\right] & = & \E  \E_0 \left[u(W_T^*)\right] =  \frac{w^{1-\gamma}}{1-\gamma} \E  \left[  e^{\gamma(A(0)+B(0)X_0+\frac{1}{2}C(0)X_0^2)} \right] \nonumber\\
&=& \frac{w^{1-\gamma}}{(1-\gamma) \sqrt{Q(0)}} e^{\gamma A(0)+\frac{\gamma}{2Q(0)}\left( \gamma B(0)^2 R_0 +
2\pi_0 B(0) + C(0) \pi_0^2 \right)}. \label{eq:gaussian}
\end{eqnarray}


To study the conditional distribution of the optimal wealth we
 compute  the conditional moment generating function of  $\ln W^*_t$,
\begin{equation}
\phi_s(t,z) := \E_s\left[(W_t^*)^z\right],\label{MGF_lnW1}
\end{equation}
on its domain of existence.

\begin{proposition}\label{prop:MGF_FI}
Let $\phi_s(t,z)<\infty$  for
$0\leq s \leq t \leq T$ and $z>0$. Then
\begin{equation}
\phi_s(t,z) = \left( \lambda_0 \xi^*_s \right)^{-\frac{z}{\gamma}}e^{D(s;t,z)+E(s;t,z)X_s+\frac{1}{2}H(s;t,z)X^2_s}
\label{MGF_lnW3}
\end{equation}
 where the functions $D:[0,t]\to \R$, $E:[0,t]\to \R$ and $H:[0,t]\to \R$ satisfy the  system of differential equations\footnote{The dependence on $t$ and $z$ for the functions
 $D,E,H$ in system \eqref{system:DEF} is omitted for  ease of notation.}
\begin{equation}\label{system:DEF}
\left\{
\begin{split}
\frac{dH}{ds} &= d(s)+2e(s)H(s)-\sigma_X^2H(s)^2\\
\frac{dE}{ds} &= f(s) + \left(e(s)-\sigma_X^2H(s)\right)E(s)+g(s)H(s)\\
\frac{dD}{ds} &= h(s) + g(s)E(s) -\frac{\sigma_X^2}{2}(H(s)+E
(s)^2)
\end{split}\right.
\end{equation}
with boundary conditions
\begin{equation}
D(t) = zA(t), \qquad E(t) = zB(t), \qquad H(t) = zC(t),\label{BC_mgf_FI}
\end{equation}
and
\begin{eqnarray}
d(s) &=& -\left(\frac{z^2}{\gamma^2}+\frac{z}{\gamma}\right)\left(1+\gamma^2\sigma_X^2(1-\rho^2)C(s)^2\right)\nonumber\\
e(s) &=& \lambda-\frac{z}{\gamma}\sigma_X\rho + z \sigma_X^2(1-\rho^2)C(s)\nonumber\\
f(s) &=& -(z^2+z\gamma)\sigma_X^2(1-\rho^2)B(s)C(s)\nonumber\\
g(s) &=& z\sigma_X^2(1-\rho^2)B(s) -\lambda \bar X\nonumber\\
h(s) &=& -\frac{1}{2}(z^2+z\gamma)\sigma_X^2(1-\rho^2)B(s)^2-\frac{zr}{\gamma}\nonumber
\end{eqnarray}
where the functions $A(\cdot), B(\cdot)$ and $C(\cdot)$ solve  \eqref{eqa}--\eqref{BC_Riccati}.
\end{proposition}

\begin{proof}

From \eqref{OptW_t_a} and the fact that process $(\ln \xi^*,X,X^2)$ is affine, it follows that
\begin{eqnarray*}
\phi_s(t,z) &=& \lambda_0^{-z/\gamma}  \E_s\left[{\xi_t^*}^{-z/\gamma}e^{zA(t)+zB(t)X_t+\frac{1}{2}zC(t)X_t^2}\right] = \lambda_0^{-z/\gamma} G(s, \xi^*_s, X_s;t,z)
\label{MGF_lnW2}
\end{eqnarray*}
where
\begin{equation}
G(s,\xi^*_s,X_s;t,z) = {\xi^*_s}^{-z/\gamma}e^{D(s;t,z)+E(s;t,z)X_s+\frac{1}{2}H(s;t,z)X^2_s}.
\label{MGF1}
\end{equation}
The boundary conditions \eqref{BC_mgf_FI} follow from $\phi_t(t,z)  = {W_t^*}^z$.
The function $G(s, \xi, x; t, z)$ is differentiable with respect to $s$, and twice differentiable with respect to $\xi$ and $x$.
Moreover, by definition, the process $(G(s, \xi_s, X_s;t,z))_{\{s\in[0,t]\}}$
is a martingale. Hence, by applying It\^{o}'s formula we get
\begin{align}
0=&\frac{\partial G}{\partial s}-r\xi\frac{\partial G}{\partial \xi}-\lambda(x-\bar X)\frac{\partial G}{\partial x}\nonumber \\
&+\frac{1}{2}\left(\frac{\partial^2 G}{\partial x^2} \sigma_X^2+ \xi^2 \frac{\partial^2 G}{\partial \xi^2} (\nu^2(1-\rho^2)+x^2)+2\xi \sigma_X \frac{\partial^2 G}{\partial \xi
\partial x}(-\nu(1-\rho^2) - x \rho)\right). \label{eq:gen0} \end{align}
By plugging  \eqref{MGF1} into Equation \eqref{eq:gen0} and collecting the constant term and the factors of  $X$ and $X^2$,  we get that $D,E,H$ are the unique solution to problem \eqref{system:DEF}-\eqref{BC_mgf_FI}  (see Filipovi\'c (2009) \cite[Lemma 10.1]{filipovic2009}).
\end{proof}


We note that the solution to problem \eqref{system:DEF}-\eqref{BC_mgf_FI} assumes a simple form in the special case corresponding to the computation of the conditional expectation of ${W^*_T}^{1-\gamma}$. In fact, from \eqref{MGF_lnW1} and \eqref{MGF_lnW3}, we get
\begin{equation*}
\E_t\left[u(W_T^*)\right] = \frac{1}{1-\gamma} \phi_t(T,1-\gamma) = \frac{1}{1-\gamma}
\left( \lambda_0 \xi^*_t \right)^{1-1/\gamma}e^{D(t;T,1-\gamma)+E(t;T,1-\gamma)X_t+\frac{1}{2}H(t;T,1-\gamma)X^2_t}.
\end{equation*}
Hence, from \eqref{phisT},
\begin{eqnarray*}
D(t;T,1-\gamma) & = & A(t), \\
E(t;T,1-\gamma) & = & B(t), \\
H(t;T,1-\gamma) & = & C(t).
\end{eqnarray*}
Such relations can also be directly verified by substituting $z=1-\gamma$ and $t=T$ in  \eqref{system:DEF}-\eqref{BC_mgf_FI}.

\section{Optimal investment under  partial information}\label{sec:partial_info}
In this section we assume that the investor observes only the stock prices and not the market price of risk. Hence, the available information is
carried by the filtration $\bF^S$ generated by the process $S$ and the  investor can only adopt $\bF^S$-adapted portfolio strategies. Here the standard procedure is to apply separability and transform the optimization problem under partial information
into an equivalent one  by means of filtering, see, e.g. Fleming and Pardoux (1982) \cite{fleming1982optimal}. The first step of this procedure consists of replacing the
unobservable quantities  by their filtered estimates. In this way, the dynamics of stock price and of the ``filtered" market price of risk  turn out to be  driven by a single, one-dimensional Brownian motion, the so called {\em Innovation process}. Hence,
after this transformation, we are in a  complete market model and, in the second step of the procedure,  we can solve the
optimization problem by following the standard martingale approach.

Let us consider the information filtration $\mathbb{F}^S:=(\F^S_t)_{t \in [0,T]}$, where, at any time $t \in [0,T]$, $\F^S_t:=\sigma\{S_u, \ 0\leq u
\leq t\}\vee \mathcal N$ and $\mathcal N$ is the collection of $\P$-null sets. We recall that $\F^S_0$ is the trivial $\sigma$-algebra.
We denote by $\pi$ the conditional expectation of $X$, given the information flow, that is $\pi_t= \E\left[X_t | \F^S_t\right],$
for every $t \in [0,T]$ and by $R$ the conditional variance, $R_t:= \E\left[\left(X_t-\E[X_t|\F^S_t]\right)^2|\F^S_t\right]$ for every $t \in
[0,T]$. It is well known that the conditional distribution of $X$  is Gaussian and hence completely identified  by the dynamics of expectation
and variance.

To characterize these dynamics we introduce the innovation process  $I=\{I_t, \ t \in
[0,T]\}$,
\begin{equation*}\label{eq:innovation}
I_t:= Z^S_t+\int_0^t (X_u-\pi_u) d u,
\end{equation*}
for every $t \in [0,T]$. Following Lipster and Shiryaev (2001) \cite[Chapter 10]{lipster2001statistics}, it can be proved that $I$ is an $(\bF^S, \P)$-Brownian motion and that the  processes $\pi$ and $R$  are the unique solutions
to the system
\begin{equation*}\label{eq:cond_exp}
d \pi_t=-\lambda (\pi_t- \bar X) d t + (R_t+\rho\sigma_X) d I_t, \quad \pi_0\in \R,
\end{equation*}
\begin{align}\label{eq:R}
d R_t = \left[\sigma_X^2 -2\lambda R_t -(R_t+\rho\sigma_X)^2\right] d t, \quad R_0\in \R^+.
\end{align}
From equation \eqref{eq:R}, we see that $R_t$ is a deterministic function of time. Therefore to emphasise this fact, from now on we will write $R(t)$ instead of $R_t$.

The semimartingale representations with respect to the information filtration $\bF^S$ of the  stock
price process and of the wealth produced by a strategy $\theta$  are
\begin{align}
\frac{d S_t}{S_t}&= (r + \sigma \pi_t) d t +  \sigma d I_t, \nonumber \\
\frac{d W_t}{W_t}&= (r + \theta_t \sigma\pi_t) d t + \theta_t \sigma dI_t.\label{eq:wealthPI}
\end{align}

The investor wants to solve the problem
\begin{equation*}\label{pr:opt_partial}
\max_{\theta \in \mathcal{A}^{S}(w)}\E\left[\frac{1}{1-\gamma}W_T^{1-\gamma}\right]
\end{equation*}
where $\mathcal{A}^{S}(w)$ is the set of $\bF^S$-predictable self-financing strategies satisfying the  integrability condition \eqref{eq:integrability} with initial wealth $w$. The state price density process in this case is unique and is given by
\begin{equation*}
\frac{d\tilde \xi_t}{\tilde \xi_t} = -r dt - \pi_t dI_t, \quad \tilde \xi_0=1.
\end{equation*}
By the martingale method we  formulate the equivalent
static problem
\begin{equation}
\max_{ W_T} \E[u( W_T)],
\label{static_pr_pi}
\end{equation}
subject to the constraint
\begin{equation}
w = \E[\tilde\xi_T W_T].
\label{static_BC_pi}
\end{equation}
We note that, since $\bF^S$-predictable strategies are also $\bF$-predictable, the optimal
utility under partial information is always lower than that under full information, and hence if problem \eqref{static_pr}-\eqref{static_BC} is bounded, problem
\eqref{static_pr_pi}-\eqref{static_BC_pi}  is also bounded.

By the usual Lagrangian approach, since $\tilde \xi$ is the state price density process, the optimal wealth satisfies
$$
\tilde W^*_t=\tilde \lambda_0^{-\frac{1}{\gamma}} \tilde{\xi}_t^{-1}\E[\tilde{\xi}_T^{1-\frac{1}{\gamma}}|\mathcal F^S_t]
$$
where $\tilde \lambda_0$ is the Lagrangian multiplier from the budget constraint \eqref{static_BC_pi}.

We can now state a verification theorem for the partial information setting.

\begin{theorem}[Verification Theorem under partial information] \label{VTpartial}
Let the function $ F( Y, \pi, t)$ solve the equation
\begin{eqnarray}
&\frac{1}{2}&  F_{ Y  Y} Y^2 \pi^2 +  F_{\pi  Y} Y \pi (R+\rho \sigma_X) +\frac{1}{2} F_{\pi \pi}(R+\rho \sigma_X)^2 +
F_t
\nonumber \\
&=&  r  F - r  F_{ Y}  Y  + F_\pi \left((R+\rho \sigma_X)\pi  + \lambda(\pi-\bar X) \right) \label{eq:pde_partial}
\end{eqnarray}
with  boundary conditions
\begin{eqnarray}
 F( Y,\pi,T) =  Y^{\frac{1}{\gamma}}, \quad and \quad  F( Y_0, \pi_0,0)= w \label{c2partial}
\end{eqnarray}
for some constant $ Y_0 > 0$ and $ F( Y,\pi,t)\rightarrow  F( Y,\pi,T)$ as $t\rightarrow T$.

Let   $ Y_t:=Y_0 \left( \tilde \xi_t\right)^{-1}$  and assume that
\begin{equation}\label{eq:c4}
\E\left[\int_0^T \left((F_{ Y}( Y_t,\pi_t,t)\pi_t)^2+( Y_t^{-1}  F_{\pi}( Y_t,\pi_t,t))^2+ ( Y_t^{-1} \pi_t  F( Y_t,\pi_t,t))^2\right) dt \right]<\infty.
\end{equation}

Then  the optimal wealth is
$\tilde W^*_t= F( Y_t,\pi_t,t)$ and  the optimal investment strategy is
\begin{align}\label{eq:optimal_strategy_p}
\tilde \theta_t^*=\frac{ F_{ Y}( Y_t, \pi_t, t) Y_t \pi_t + (R(t)+\rho \sigma_X)  F_{\pi}( Y_t,
\pi_t, t)}{\sigma  F(Y_t, \pi_t, t)},
\end{align}
for all $t \in [0,T]$.
\end{theorem}

\begin{proof}
Similarly to the proof of Theorem \ref{VT}, we need to show that the initial wealth satisfies the budget constraint  and that the
final wealth satisfies the first order condition 
and is attainable by  a self-financing strategy.

The budget constraint 
and the first order condition follow from \eqref{c2partial}. By It\^o formula we get
\begin{align*}
\frac{d  Y_t}{ Y_t}=&(r+ \pi^2_t) d t + \pi_t dI_t.
\end{align*}
Hence,
\begin{align}
& F( Y_t, \pi_t, t)= F( Y_0, \pi_0, 0)+ \int_0^t\widetilde{\mathcal{L}} F( Y_s, \pi_s, s) d s \nonumber \\
&\quad +\int_0^t\left((R_t+\rho \sigma_X)  F_\pi( Y_s, \pi_s, s)+ Y_s \pi_s  F_{ Y}( Y_s, \pi_s, s)\right) d I_s  \label{eqFpart}
\end{align}
where $\widetilde{\mathcal{L}}$ is the differential operator
\begin{align*}
\widetilde{\mathcal{L}}  F=&  F_t +\frac{1}{2}F_{ Y Y} Y^2 \pi^2  +  F_{\pi  Y} Y \pi (R_t+\rho \sigma_X)
\nonumber \\
&+ \frac{1}{2} F_{\pi \pi}(R_t+\rho \sigma_X)^2  -  F_\pi  \lambda(\pi-\bar X)+   F_{ Y}  Y (r+\pi^2).\nonumber
\end{align*}
To show that the optimal wealth can be obtained  by a self-financing strategy starting from $w$
it remains to prove that  the process $\tilde \xi_t  F( Y_t,\pi_t, t)$ is a true martingale.
By applying the product rule and using  \eqref{eq:pde_partial}, we get
\begin{align*}
&\tilde \xi_t  F( Y_t, \pi_t, t)=F( Y_0, \pi_0, 0)\\
&+ \int_0^t\tilde \xi_s\left((R_t+\rho \sigma_X)  F_\pi( Y_s, \pi_s, s) + Y_s \pi_s  F_{Y}(Y_s, \pi_s, s)-  F( Y_s, \pi_s, s) \pi_s\right) d I_s.
\end{align*}
By \eqref{eq:c4} and the fact that $R_t$ is the solution to the Riccati equation \eqref{eq:R} on $[0,T]$,  we get that  the integral with respect to $I$ is a true martingale. Then $\tilde W^*_t= F( Y_t, \pi_t, t)$ is the optimal wealth process and the optimal investment strategy in \eqref{eq:optimal_strategy_p} is obtained by equating the predictable covariation processes with respect to $I$ from \eqref{eq:wealthPI} and \eqref{eqFpart}.

\end{proof}


To obtain a closed form representation for the optimal wealth we guess that
$$
\E[\tilde{\xi}_T^{1-\frac{1}{\gamma}}|\mathcal F^S_t]=\tilde{\xi}_t^{1-\frac{1}{\gamma}} e^{\tilde A(t)+\tilde B(t)\pi_t+\frac{1}{2}\tilde C(t)\pi_t^2}
$$
where the functions $\tilde A(t), \tilde B(t)$ and $\tilde C(t)$ satisfy the system of Riccati Equations
\begin{equation}
\label{eqa_pi}
\left\{\begin{split}
\frac{d\tilde C}{dt} &= \tilde a + \tilde b(t) \tilde C(t) + \tilde c(t) \tilde C(t)^2,\\
\frac{d\tilde B}{dt} &=  -\tilde C(t)\lambda\bar{X} + \left(\frac{\tilde b(t)}{2}+\tilde c(t) \tilde C(t)\right)\tilde B(t), \\
\frac{d\tilde A}{dt} &= \frac{\gamma-1}{\gamma}r - \tilde B(t)\lambda\bar{X} + \frac{1}{2}\tilde c(t)\left( \tilde B(t)^2 + \tilde
C(t)\right)
\end{split}\right.
\end{equation}
with boundary conditions
\begin{equation}
\tilde A(T) = \tilde B(T) = \tilde C(T) = 0 ,
\label{BC_Riccati_pi}
\end{equation}
where
\begin{align*}
\tilde a =\frac{\gamma-1}{\gamma^2}, \quad \tilde b(t) = 2\left(\lambda +\frac{\gamma-1}{\gamma}\left(R(t)+\rho\sigma_X\right)\right), \quad \tilde c(t) =
-\left(R(t)+\rho\sigma_X\right)^2.
\end{align*}

The solutions to the non-homogeneous system of Riccati equations \eqref{eqa_pi}-\eqref{BC_Riccati_pi} are related to the solutions of the homogeneous system  \eqref{eqa}-\eqref{BC_Riccati} arising in the full information case.
This fact, shown in the next proposition, will be exploited to get  simpler expressions for many quantities of  interest.

\begin{proposition}\label{prop:Q}
Let the pairs of functions $B(t), C(t)$  and $\tilde B(t), \tilde C(t)$  satisfy the problems  \eqref{eqa}-\eqref{BC_Riccati} and  \eqref{eqa_pi}-\eqref{BC_Riccati_pi} on $[0,T]$, respectively and let
$$Q(t):=1-\gamma C(t) R(t).$$

Then, for all $t$ in $[0,T]$, $Q(t)$ is strictly positive and
\begin{eqnarray}
\tilde{C}(t) = Q(t)^{-1}C(t), \label{tildeC} \\
\tilde{B}(t) = Q(t)^{-1}B(t). \label{tildeB}
\end{eqnarray}
Moreover, the functions $C(t)$ and $\tilde C(t)$ are strictly positive and decreasing on [0,T] if $\gamma <1$ and are strictly negative and increasing if $\gamma > 1$.
\end{proposition}

\begin{proof}
The fact that the function $C(t)$ is strictly positive and decreasing on [0,T] if $\gamma <1$ and it is negative and increasing for $\gamma > 1$ has been proven by Kim and Omberg (1996) \cite[Equation (23)]{kim}.

The function $Q(t)$ is continuous, hence the set  $\mathcal{T}:=\left\{t \in [0,T] \vert Q(t) = 0 \right\}$ is closed; we want to show that it is empty. By contradiction, let us assume that it is not empty and let $\bar t$ be its maximum. From the boundary condition \eqref{BC_Riccati} we see that $Q(T)=1$, hence $\bar t < T$. Relations \eqref{tildeC} and \eqref{tildeB} hold in the set $\mathcal{T}^C \cap [0,T]$, where $\mathcal{T}^C$ is the complement of $\mathcal{T}$. In fact they follow from the fact that $Q(t)^{-1}C(t)$ and $Q(t)^{-1}B(t)$ satisfy \eqref{eqa_pi}-\eqref{BC_Riccati_pi} when $C(t)$, $B(t)$ satisfy \eqref{eqa} -\eqref{BC_Riccati},
as it can be shown by following Brendle (2006) \cite[Equations (28)-(29)]{brendle2006}.
Therefore, for any $\epsilon > 0$ such that $\bar t + \epsilon <T$, $Q(\bar t+\epsilon)\tilde C(\bar t+\epsilon) = C(\bar t+\epsilon)$ and, by continuity of all the functions involved in the equality,  $Q(\bar t)\tilde C(\bar t) = C(\bar t)$.
Since $C(t)$ is a monotone function (either increasing or decreasing, depending on the parameter $\gamma$) and $C(T)=0$, then $C(\bar t) \not= 0$, hence $\bar t \notin \mathcal{T}$, that is a contradiction and $\mathcal{T}$ is the empty set.

Since $\mathcal{T}$ is empty, $Q(t)$ is continuous on $[0,T]$ and $C(T)=1$, it follows that $Q(t)$ is strictly positive on $[0,T]$, hence the functions $C(t)$ and $\tilde{C}(t)$ must have the same sign (positive for $\gamma <1$ and negative for $\gamma >1$).

Finally, we prove that for $\gamma < 1$, $\tilde{C}(t)$ is strictly decreasing on $[0,T]$.
Consider the equation
$$
\frac{d\tilde C(t)}{dt} = f(\tilde C(t)),
$$
where $f(\tilde C(t))$ is the right hand side of the first equation in \eqref{eqa_pi}-\eqref{BC_Riccati_pi}.
The boundary condition implies that $\tilde C(T) = 0$ and that $f(0) = \frac{\gamma-1}{\gamma^2} < 0$. Then the function $f(t)$ must be negative on $[0,T]$  for the boundary condition to be satisfied and hence $\tilde C(t)$ is strictly decreasing.
The same argument applies to the case  $\gamma >1$ where the derivative of $\tilde C(t)$ is positive and hence $\tilde{C}(t)$ is strictly increasing.
\end{proof}

We remark that from \eqref{tildeC}-\eqref{tildeB},  we can get an explicit expression for  $\tilde{B}(t)$ and $\tilde{C}(t)$ from those of $B(t)$ and $C(t)$. Then $\tilde{A}(t)$ can be obtained explicitly by integrating the right hand side of the third equation in system \eqref{eqa_pi}-\eqref{BC_Riccati_pi}.

We are now ready to determine the optimal wealth and the optimal investment strategy for the partial information problem.
\begin{theorem}\label{prop:OptSolPartInfo}
Let the functions $\tilde A(t)$, $\tilde B(t)$ and $\tilde C(t)$ satisfy \eqref{eqa_pi}--\eqref{BC_Riccati_pi} on $[0,T]$ and let
\begin{equation*}
\tilde \lambda_0 = \left[\frac{e^{\tilde A(0)+\tilde B(0) \pi_0+\frac{1}{2}\tilde C(0) \pi_0^2}}{w}\right]^\gamma.
\label{K_pi}
\end{equation*}
Assume that
\begin{equation}
\E\left[\int_0^T \left(\tilde \xi_t^{1-\frac{1}{\gamma}}e^{\tilde A(t)+\tilde B(t)\pi_t+\frac{1}{2}\tilde C(t)\pi_t^2}\right)^2 (1+\pi_t^2)     \right]<\infty .\label{eq:integral4}
\end{equation}
Then
$\tilde \lambda_0$ is the Lagrangian multiplier from the budget constraint \eqref{static_BC_pi} and the optimal wealth and the optimal investment strategy are given by
\begin{align}
\tilde W_t^* = (\tilde \lambda_0 \tilde \xi_t)^{-\frac{1}{\gamma}} e^{\tilde A(t)+\tilde B(t)\pi_t+\frac{1}{2}\tilde C(t)\pi_t^2},
\label{OptW_t_pi}\\
\tilde\theta_t^* = \frac{1}{\gamma}\frac{\pi_t}{\sigma} +\frac{(R(t)+\rho\sigma_X)}{\sigma}(\tilde B(t)+\tilde C(t)\pi_t),
\nonumber
\end{align}
for every $t \in [0,T]$.
\end{theorem}

\begin{proof}
The proof follows from the same argument of the analogous result under full information, Theorem \ref{prop:OptSolFullInfo}, and hence is omitted.
\end{proof}

In the next proposition we provide sufficient conditions to apply Theorem \ref{prop:OptSolPartInfo} that are easier to check for a given set of parameters.

\begin{proposition}\label{prop:boundPI}
Let the functions $\tilde A(t)$, $\tilde B(t)$ and $\tilde C(t)$ satisfy \eqref{eqa_pi}--\eqref{BC_Riccati_pi} on $[0,T]$ and
assume that at least one of the following two holds
\begin{enumerate}
\item[(i)] $\gamma>1$
\item[(ii)] The functions $A(t)$, $B(t)$ and $C(t)$ satisfy \eqref{eqa}- \eqref{BC_Riccati} on $[0,T]$ and
\begin{equation}
1-4\frac{C(0)}{Q(0)}\max \left(R_0, R_0e^{-2\lambda T}+\frac{\sigma_X^2}{2\lambda}(1-e^{-2\lambda T})\right)>0.
\label{eq:condPI}
\end{equation}
\end{enumerate}

Then all assumptions of  Theorem \ref{prop:OptSolPartInfo} are satisfied.
\end{proposition}

\begin{proof}
We only need to show that the integrability condition  \eqref{eq:integral4} in Theorem \ref{prop:OptSolPartInfo} is satisfied.

Using Fubini and the Cauchy Schwartz inequality we get
\begin{align*}
&\E\left[\int_0^T e^{2\tilde A(t)} \left(\tilde{\xi}_t^{1-\frac{1}{\gamma}}\right)^2 e^{2\tilde B(t)\pi_t+\tilde C(t)\pi_t^2} (1+\pi_t^2) dt  \right]\\
&\quad \leq \kappa_1 \int_0^T e^{2\tilde A(t)} \E\left[\tilde{\xi}_t^{8(1-\frac{1}{\gamma})}\right]^{\frac{1}{4}} \E\left[e^{4 \tilde B(t)\pi_t+2\tilde C(t)\pi_t^2}\right]^{\frac{1}{2}} \E[(1+\pi_t^2)^4]^{\frac{1}{4}} dt.
\end{align*}
Since $\pi_t$ is Gaussian,
$
\E[(1+\pi_t^2)^4]<\infty
$.
The expectation
$\E\left[\tilde{\xi}_t^{8(1-\frac{1}{\gamma})}\right]$
is finite since $\pi$ is Ornstein-Uhlenbeck (see, Revuz and Yor (2013)\cite[Chaper 8, Ex. 3.14]{revuz}).
Finally, $\E\left[e^{4\tilde B(t)\pi_t+2\tilde C(t)\pi_t^2}\right]$ is finite for all $t \in [0,T]$ if and only if
$1-4\tilde C(t)\tilde{v}_t>0$ where
$\tilde v_t = v_t - R_t$ is the variance of $\pi_t$ (and $v_t$ is the variance of $X_t$).

If  $\gamma >1$,  from Proposition \ref{prop:Q},   $\tilde{C}(t)<0$. Hence $1-4\tilde C(t)\tilde{v}_t >0$ and \eqref{eq:integral4} is satisfied.

 If $\gamma < 1$, still from Proposition \ref{prop:Q} $\tilde{C}(t)$ is strictly positive and  decreasing in $[0,T]$.
Therefore
\begin{eqnarray}
1-4\tilde C(t)\tilde{v}_t &>& 1-4\tilde C(0)v_t \geq 1-4\frac{C(0)}{Q(0)}\max \left(R_0, R_0e^{-2\lambda T}+\frac{\sigma_X^2}{2\lambda}(1-e^{-2\lambda T})\right),\nonumber
\end{eqnarray}
where the first inequality follows from the monotonicity of $\tilde{C}$ and from the fact that $\tilde{v}_t < v_t$. The second inequality follows from \eqref{tildeC} and from the fact that $v_t$ is always lower than its maximum value on $[0,T]$ that is equal to $R_0$ or to $v_T$  depending on $R(t)$ being decreasing  or increasing. Then the result follows immediately from \eqref{eq:condPI}.
\end{proof}

Now we can compute the conditional moment generating function of the optimal wealth under the partial information,
$$\tilde{\phi}_s(t,z) :=  \E\left[(\tilde{W}^*_t)^z|\mathcal{F}^S_s\right]. $$

\begin{proposition}\label{prop:MGF_PI}

Let $\tilde{\phi}_s(t,z)<\infty$ for $0\le s \le t \le T$
and  $z>0$.

Then
\begin{equation*}
\tilde{\phi}_s(t,z) = (\tilde \lambda_0 \tilde \xi_s)^{-\frac{z}{\gamma}}e^{\tilde{D}(s;t,z)+\tilde{E}(s;t,z)\pi_s+\frac{1}{2}\tilde{H}(s;t,z)\pi^2_s}
\label{MGF_lnW3_PI}
\end{equation*}
where   $\tilde{D}:[0,t]\to \R$, $\tilde{E}:[0,t]\to \R$ and $\tilde{H}:[0,t]\to \R$ satisfy the  system of differential equations\footnote{
 Note that the functions $\tilde D(s), \tilde E(s)$ and $\tilde H(s)$  depend on $t$ and $z$. We do not report such dependence into the formulas for a simpler notation.}
\begin{equation*}\label{system:DEF_PI}
\left\{
\begin{split}
\frac{d\tilde{H}}{ds} &= \tilde{d}(s)+2\tilde{e}(s)\tilde{F}(s)+\tilde{f}(s)\tilde H(s)^2\\
\frac{d\tilde{E}}{ds} &= \left(\tilde{e}(s)+\tilde{f}(s)\tilde{H}(s)\right)\tilde{E}(s)-\lambda \bar{X}\tilde{H}(s)\\
\frac{d\tilde{D}}{ds} &= -\frac{zr}{\gamma} -\lambda\bar{X}\tilde{E}(s) + \frac{1}{2}\tilde{f}(s)(\tilde{H}(s)+\tilde{E}(s)^2)
\end{split}\right.
\end{equation*}
with boundary conditions
\begin{equation}
\tilde{D}(t) = z\tilde{A}(t), \qquad \tilde{E}(t) = z \tilde{B}(t), \qquad \tilde{H}(t) = z\tilde{C}(t),\label{BC_mgf_PI}
\end{equation}
and
\begin{eqnarray*}
\tilde{d}(s) &=& -\left(\frac{z^2}{\gamma^2}+\frac{z}{\gamma}\right),\nonumber\\
\tilde{e}(s) &=& \lambda-\frac{z}{\gamma}(R(s) + \rho\sigma_X), \nonumber\\
\tilde{f}(s) &=& -(R(s)+\rho\sigma_X)^2,\nonumber
\end{eqnarray*}
for every $s \leq t$, and where the functions $\tilde{A}$, $\tilde{B}$ and $\tilde{C}$ satisfy \eqref{eqa_pi}-\eqref{BC_Riccati_pi}.
\end{proposition}

\begin{proof}
The proof replicates the  steps of the proof of  Proposition \ref{prop:MGF_FI}. Using that the process $(\ln \tilde \xi, \pi, \pi^2)$ is affine
we have
\begin{eqnarray*}
\tilde\phi_s(t,z) &=& \tilde\lambda_0^{-z/\gamma}  \E\left[\tilde{\xi_t}^{-z/\gamma}e^{z\tilde A(t)+z\tilde B(t)\pi_t+\frac{1}{2}z\tilde C(t)\pi_t^2}|\mathcal{F}_s^S\right]\nonumber\\
&=&\tilde\lambda_0^{-z/\gamma} \tilde G(s, \tilde\xi_s, \pi_s;t,z)
\label{MGF_lnW2_PI}
\end{eqnarray*}
where
\begin{equation*}
\tilde G(s,\tilde\xi_s,\pi_s;t,z) = \tilde{\xi_s}^{-z/\gamma}e^{\tilde D(s;t,z)+\tilde E(s;t,z)\pi_s+\frac{1}{2}\tilde H(s;t,z)\pi^2_s}.
\label{MGF1_PI}
\end{equation*}
The boundary conditions \eqref{BC_mgf_PI} follow from $\tilde{\phi}_t(t,z)  = (\tilde{W}_t^*)^z$.
The function $\tilde G(s, \tilde\xi, \pi; t, z)$ is differentiable with respect to $s$, and twice differentiable with respect to $\tilde \xi$ and $\pi$.
Moreover, by definition, the process $(\tilde G(s, \tilde \xi_s, \pi_s;t,z))_{\{s\in[0,t]\}}$
is a martingale with respect to filtration $\bF^S$ . Hence, by applying It\^{o}'s formula we get that the function $\tilde G$ satisfies the equation
\begin{eqnarray}
0 & = &\frac{\partial \tilde{G}}{\partial s}-r\tilde\xi_s\frac{\partial \tilde{G}}{\partial \tilde\xi}-\lambda(\pi_s-\bar X)\frac{\partial
\tilde{G}}{\partial \pi}\nonumber\\
&+&\frac{1}{2}\left(\frac{\partial^2 \tilde{G}}{\partial \pi^2} (R_s+\rho\sigma_X)^2+ {\tilde\xi_s}^2\pi_s^2 \frac{\partial^2 \tilde{G}}{\partial \tilde\xi^2}
-2\tilde\xi_s\pi_s
\frac{\partial^2 \tilde{G}}{\partial \tilde\xi \partial \pi}(R_s+\rho\sigma_X) \right). \nonumber
\end{eqnarray}
This completes the proof.
\end{proof}

Note that, also under partial information, formulas simplify when
 $t= T$ and $z = 1-\gamma$, in fact:
\begin{eqnarray*}
\tilde D(s;T,1-\gamma) & = & \tilde A(s),\\
\tilde E(s;T,1-\gamma) & = & \tilde B(s),\\
\tilde H(s;T,1-\gamma) & = & \tilde C(s).
\end{eqnarray*}
This allows to compute the optimal expected utility in closed form, since
\begin{eqnarray}
\E\left[u(\tilde W_T^*)|\mathcal{F}^S_s\right] &=& \frac{1}{1-\gamma} \tilde \phi_s(T,1-\gamma)\nonumber\\
&=& \frac{1}{1-\gamma} {(\tilde{\lambda}_0\tilde{\xi}_s)}^{1-1/\gamma}
e^{\tilde A(s)+ Q(s)^{-1}B(s)\pi_s +\frac{1}{2}Q(s)^{-1}C(s)\pi_s^2}
\nonumber\\
&=& \frac{1}{1-\gamma}  (\tilde{W}^*_s)^{1-{\gamma}} e^{\gamma(\tilde A(s)+ Q(s)^{-1}B(s)\pi_s+\frac{1}{2}Q(s)^{-1}C(s)\pi_s^2)} \label{phisTtilde}
\end{eqnarray}
where we have used the explicit expression of $\tilde C$ and $\tilde B$ in terms of $C$ and $B$, given in  \eqref{tildeC}-\eqref{tildeB} and where the last equality is obtained from  \eqref{OptW_t_pi}.

\section{The value of information}\label{sec:value}
We are now ready to define the {\em value of information}, that is  to assign a monetary value to the possibility of improving the knowledge of the market price of risk. We start by computing the reservation price, that is  the maximal amount of money that
a partially informed investor would be willing to pay to get extra information. We will  focus on two kinds of  information,  which we call {\em initial  } and  {\em dynamic}.  While the initial information  gives the exact knowledge of $X_0$ that is the value of the market price of risk at time $0$,  the dynamic information  provides the run-time values  $X_t$ at all times $t\in [0,T]$.

A partially informed investor, endowed with a  starting wealth $w$, with  a prior $X_0 \sim N(\pi_0,R_0)$  obtains, at time $T$, the final (optimal) wealth $\tilde W_T^*(w)$. Let us now assume that the value assumed by $X_0$ is revealed to the investor at time $0$. Then she will be able to implement the optimal strategy, still under partial information because the following path of $X$ will remain unknown to her, but this time starting from the exact value $X_0$. Let us denote by $\tilde
W_T^I(w)$ the optimal wealth obtained at time $T$, where the index $I$ highlights the {\em Initial} information case.
When dynamic information is provided to the investor, she will reach the wealth produced at time $T$ by the optimal strategy under full information, that is $W_T^*(w)$.
 Since the sets of feasible strategies for the three scenarios are strictly increasing, the following inequalities hold
\begin{align}
\E \left[u ( \tilde W_T^*(w))\right] 
& \le  \E\left[ u ( \tilde W_T^I(w))  \right]  \label{EU1} \\
& \le  \E\left[ u (W_T^*(w) ) \right].    \label{EU2}
\end{align}
The maximum amount that the investor is willing to pay to receive the initial information is the quantity $\Delta w< w$ that satisfies
\begin{equation}
\E  [u ( \tilde W_T^*(w) )]= \E\left[ u ( \tilde W_T^I(w-\Delta w)) \right] \label{EU3}.
\end{equation}
Notice that $\Delta w >0 $ because of  \eqref{EU1}  and the fact that the expected utility is increasing with respect to  the initial wealth.
From \eqref{phisTtilde} computed for $s = 0$ we get
\begin{equation}
\E \left[u ( \tilde W_T^*(w) )\right]
= \frac{w^{1-\gamma}}{1-\gamma}e^{\gamma\left(\tilde A_0(R_0) + Q_0^{-1} B_0\pi_0 +\frac{1}{2} Q_0^{-1} C_0\pi_0^2 \right)}
\label{U_PI}
\end{equation}
where we use the notation $Q_0:=Q(0)$, $B_0:=B(0)$, $C_0:=C(0)$ and $\tilde A_0(R_0):=\tilde A(0)$ to highlight the  dependence of $\tilde A(0)$  on  $R_0$.
With an analogous computation, setting $\pi_0=X_0$ and $R_0=0$, we get
 \begin{align*}
\E\left[ u ( \tilde W_T^I(w-\Delta w)) \vert X_0\right]=
 \frac{(w-\Delta w)^{1-\gamma}}{1-\gamma} e^{\gamma\left(\tilde A_0(0) + B_0X_0+\frac{1}{2}C_0X_0^2\right)}.
\end{align*}
Hence, the right hand side of \eqref{EU3} is
 \begin{equation}
 \E\left[  \E\left[ u ( \tilde W_T^I(w-\Delta w)) \vert X_0\right]\right]
 = \frac{(w-\Delta w)^{1-\gamma}}{(1-\gamma)\sqrt{Q_0}}  e^{\gamma\tilde A_0(0) + \frac{\gamma}{2Q_0}\left(\gamma B_0^2R_0 +
 2B_0\pi_0 +  C_0\pi_0^2\right)}\label{eq:U_HI}
 \end{equation}
which holds when $Q_0 = 1-\gamma C_0 R_0 > 0$. 
Solving equation \eqref{EU3} using the explicit expressions in \eqref{U_PI} and \eqref{eq:U_HI} we get $\Delta w$ and we can define the {\em Value of Initial Information} ${\mathcal V}^{I}$ as the ratio $\Delta w/w$, that is
\begin{equation}
{\mathcal V}^{I}= 1 - \left(\sqrt{Q_0}
\, e^{\gamma(\tilde A_0(R_0)-\tilde A_0(0))-\frac{\gamma^2B_0^2R_0}{2Q_0}}\right)^\frac{1}{1-\gamma}
\label{InitialInfo}.\end{equation}
We remark that ${\mathcal V}^{I}$ does not depend on the expected value of the market price of risk $\pi_0$ but only on the variance of
the initial estimate $R_0$.

Let us now compute the reservation price for the dynamic information. Let the quantity  $\Delta w$ be the solution to the equation
\begin{equation}
\E \left[u ( \tilde W_T^*(w) )\right]= \E \left[ u (W_T^*(w-\Delta w)) \right]. \label{EU4}
\end{equation}
Inequality \eqref{EU2} implies that $0<\Delta w<w$. To compute the right hand side of \eqref{EU4} we use equation\eqref{eq:gaussian}. 
The left hand side of \eqref{EU4} is given in \eqref{phisTtilde} for $s=0$.
Again we shorten the notation by using $A_0=A(0), B_0=B(0), C_0=C(0), Q_0=Q(0)$ and $\tilde{A}_0(R_0) = \tilde{A}(0)$.
Hence, we can extract the reservation price $\Delta w$ from \eqref{EU4} and define the {\em Value of Dynamic Information} ${\mathcal V}^{D}$ as the ratio $\Delta w/w$, that
is
\begin{equation}
{\mathcal V}^{D} =  1 - \left(\sqrt{Q_0} \; e^{\gamma (\tilde A_0(R_0) - A_0) - \frac{\gamma^2B_0^2R_0}{2Q_0}}\right)^{\frac{1}{1-\gamma}} \ .
\label{valueInfo}
\end{equation}
From inequalities \eqref{EU1}- \eqref{EU2} we get
\begin{equation}
 0 < {\mathcal V}^{I}  \le {\mathcal V}^{D}  < 1. \label{VIVD}
 \end{equation}
We remark that the expression for ${\mathcal V}^{D}$ can be obtained from that of ${\mathcal V}^{I} $  \eqref{InitialInfo}
by replacing  $\tilde A_0(0)$ with $A_0$. We also note that ${\mathcal V}^{D}$does not depend on the expected value of the market price of risk.

\section{Applications}\label{sec:numerics}
In this section we discuss some applications of our results with the parameters of Table \ref{Table1} obtained from the estimates provided by  Xia (2001)
\cite[Table I]{xia2001} on the U.S. stock market, from 1950 to 1997.

\begin{table}[htpb]
\begin{tabular}{c c c c c c c c c c c }
\hline\\
  $r$&$\sigma$&$\lambda$&$\sigma_X$&$\bar X$&$\pi_0$&$R_0$&$S_0$&$W_0$&$T$& $\gamma$\\
 \hline\\
  3.4\% & 14.4\% & 0.19 & 18.75\% & 0.3958 &  0.3958&0.09&1  &1&5 &5\\
\\
\hline
\end{tabular}
\caption{Parameter set adopted in this section (expressed on a yearly basis and derived from \cite[Table I]{xia2001})}
\label{Table1}
\end{table}

Figure \ref{nu2} shows  the optimal penalization factor  $\nu^*_0$ for the incomplete market under full information derived in  \eqref{nu_1}, as a function of the risk aversion parameter $\gamma$, for three different
values of the correlation between the stock price process and the market price of risk and assuming $X_0=\pi_0$. We see that  $\nu^*_0$  grows in absolute value as $\gamma$ tends to zero. This is explained by the fact that  investors with smaller
risk-aversion need a higher penalization factor (i.e. greater in absolute value) to be diverted from unattainable claims.  The case $\gamma=1$  corresponds to
logarithmic utility. Here no penalty is necessary: investor is myopic and selects only  attainable claims (see Remark 1).  For $\gamma$ larger than $1$
the size of  $\nu^*_0$ is first increasing and then decreasing  towards zero.
In fact, for values of $\gamma$ slightly larger than $1$,  investors are less myopic and attracted by not marketed claims. When investors are more risk averse they put a larger part of their wealth in the risk-free asset, and hence the penalization becomes again less necessary.

\begin{figure}[htpb]
{\includegraphics[height = 6cm, width = 12cm]{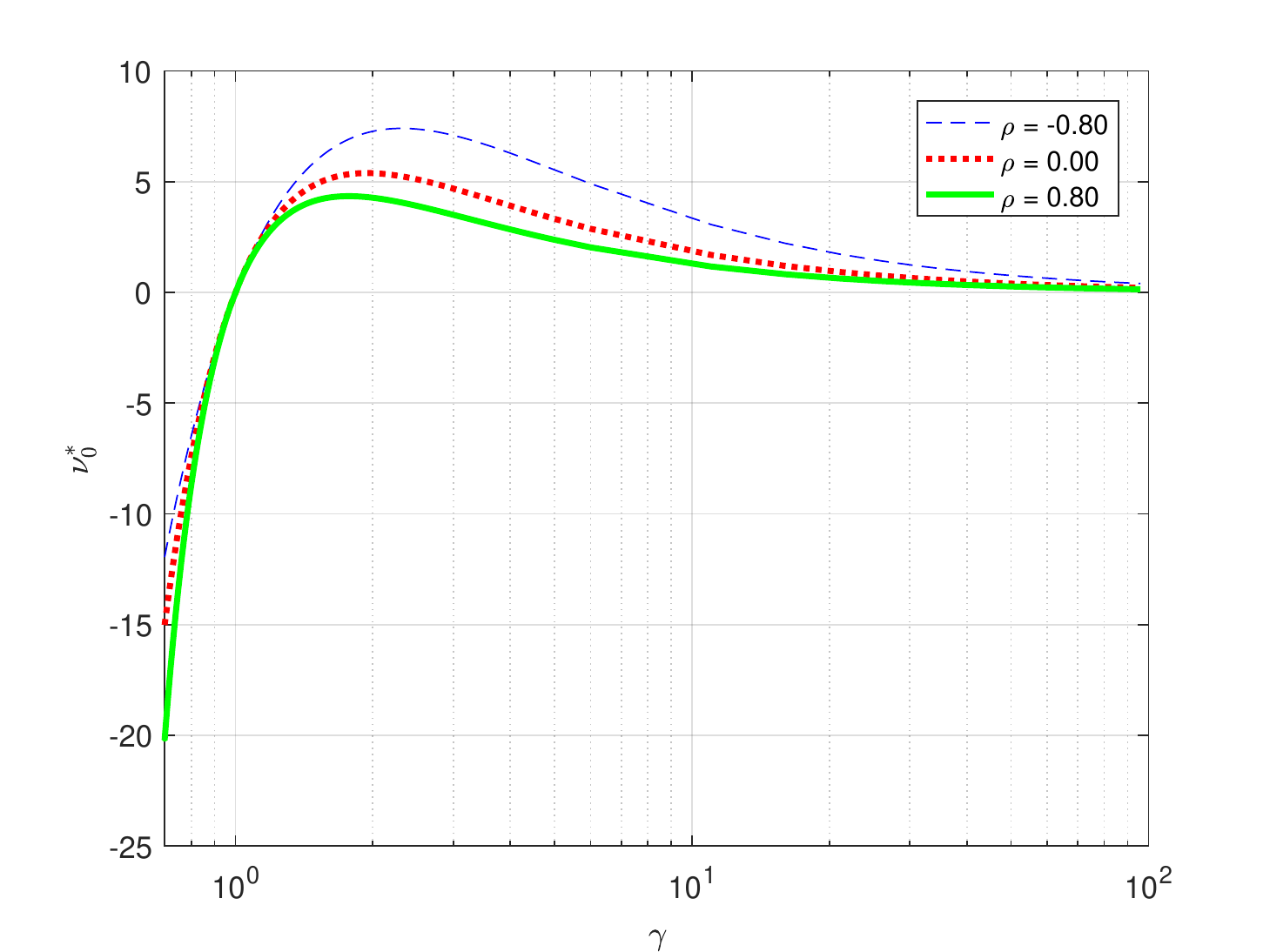}}
\caption{The optimal penalization factor $\nu^*_0$, Equation \eqref{nu_1}, as a function of the risk aversion parameter $\gamma$, when $X_0=\pi_0$, for  correlations: $\rho = 0.8$, continuous line; $\rho =
0$ dotted line; $\rho = -0.8$ dashed line.}
\label{nu2}
\end{figure}

Figure \ref{Criticaltime}  represents the critical time $T^*$, given by
\eqref{crit_t} , that is the maximal horizon of existence for
the solution to \eqref{eqa}-\eqref{BC_Riccati}, as a function of $\gamma$ for two values of
the correlation $\rho$. The analysis of existence of the 
system of Riccati equations states that $T^*$  is finite  when  $\rho$ is larger than $\rho^* \simeq -0.4934$ given by  \eqref{rho_crit},  and for values of $\gamma$ smaller than the value $\gamma^*$ defined by \eqref{gamma_crit}. In this case  $\gamma^*$ is equal to
$0.4933$ when $\rho = 0$, and to $0.7185$ when $\rho = 0.8$.  When $\rho = 0$ Figure \ref{Criticaltime} (left panel) shows that the critical time corresponding to $\gamma = 0.4$ is
about 20 years and it becomes larger than 20 years for $\gamma > 0.4$. In other words, for  $\gamma > 0.4$ the solution to the system
\eqref{eqa}-\eqref{BC_Riccati} is well defined  up to an investment horizon of at least 20 years, and it is well defined for any horizon when  $\gamma > \gamma^*$. In the same plot we also report $T^{**}$ which is the maximal time such that \eqref{eq:cond0} is satisfied. Remind that  \eqref{eq:cond0} is a sufficient condition for the optimal wealth under full information $W^*_t$ to be expressed as in \eqref{OptW_t_a}. When $\rho > \rho^*$  and $\gamma < \gamma^*$ there is a large region in the plane $\gamma,T$ where the solution to the Riccati system if well defined but
\eqref{eq:cond0} does not hold, hence to state that formula \eqref{OptW_t_a} provides the optimal wealth, one should
 prove that  the more general condition \eqref{eq:integral2} of Theorem \ref{prop:OptSolFullInfo} holds.

\begin{figure}[htpb]
{\includegraphics[height = 6cm, width = 12cm]{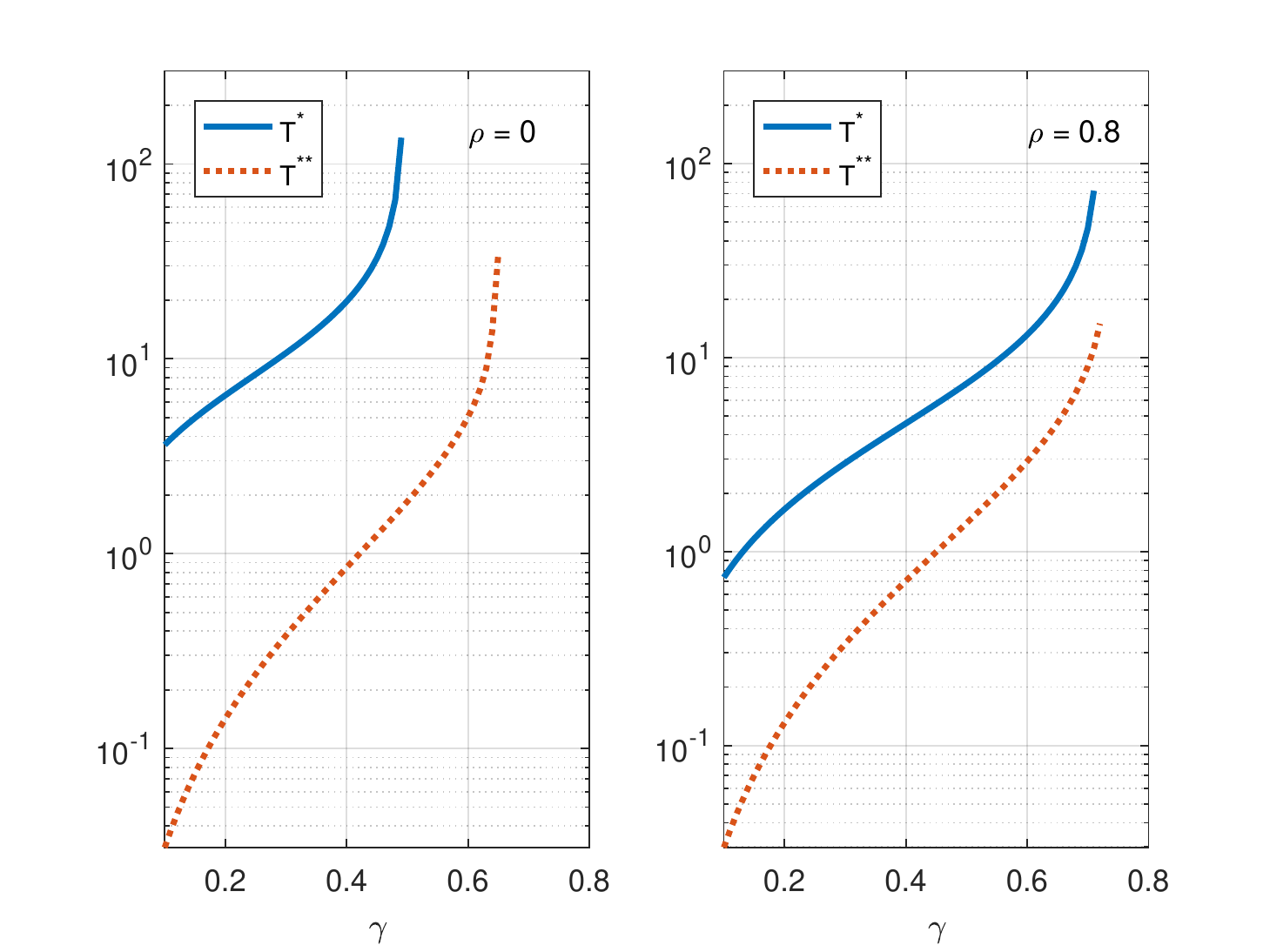}}
\caption{Critical times $T^*$ (continuous line), Equation \eqref{crit_t}, for the system of Riccati equations \eqref{eqa}-\eqref{BC_Riccati}, as a function of $\gamma$, for  correlations $\rho = 0$ (left panel) and $\rho = 0.8$ (right panel). Superimposed the maximal time $T^{**}$ (dotted line) such that \eqref{eq:cond0} is satisfied.}
\label{Criticaltime}
\end{figure}

Propositions \ref{prop:MGF_FI} and \ref{prop:MGF_PI}  characterize the moment generating functions of the optimal wealth under full or partial information. Applying those results and  Fast Fourier Transform we can  compute the corresponding probability distributions very efficiently.  Figure \ref{pdf_bis} represents the probability density functions of the optimal wealth in $T$ for a fully informed investor with $\gamma = 4.03$ and for a partially informed one with $\gamma=2.08$.  We also plot the empirical distributions, obtained by simulations, for a visual check of the precision of our code. The values of $\gamma$ have been chosen so that the expected returns of the two strategies are equal to $15\%$.  Albeit with the same mean, the two distributions have  very different shapes, with the full information density being  more skewed and with a heavier right tail. This has interesting consequence on
the mean-variance curves corresponding to the full and the partial information investment strategies for different level of risk aversions, represented in Figure
\ref{Eff_frontier}. To connect Figure \ref{Eff_frontier} and Figure \ref{pdf_bis} we also indicate the points corresponding to the two values of $\gamma$ for which we computed the densities. We see that the curve of expected returns under partial information dominates the full information one. Hence, if an investor following a mean-variance criterion (as, for instance, maximizing the Sharpe ratio of her investment) had to choose between optimal strategies under full or under partial information, she would always select the partial information one. This is a consequence of the heavier right tail of the wealth distribution under full information (clearly shown in Figure \ref{pdf_bis}), a feature not much appreciated  by a  mean-variance kind of investor.

\begin{figure}[htpb]
{\includegraphics[height = 6cm, width = 12cm]{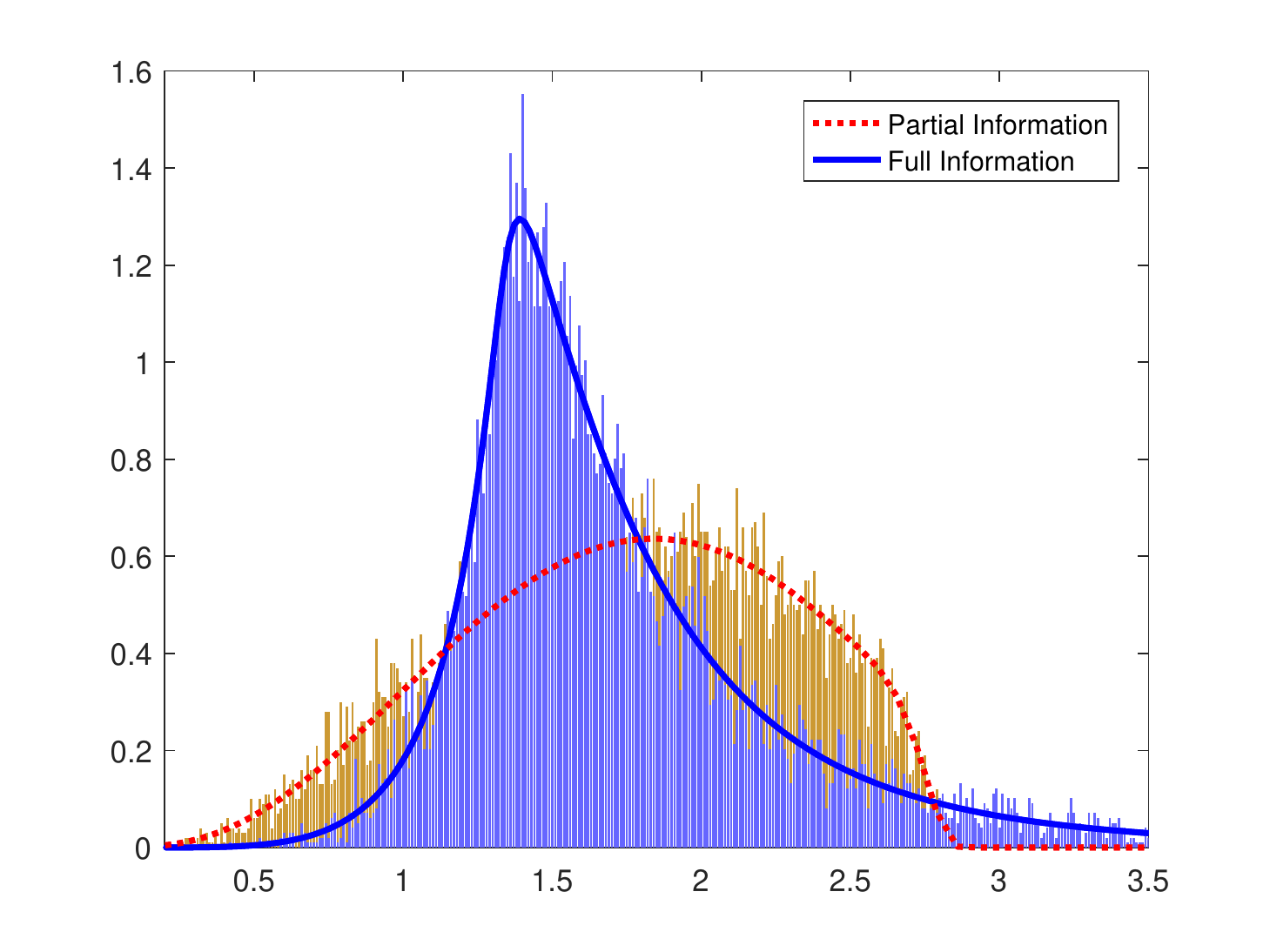}}
\caption{Probability distribution for the optimal final wealth under full (continuous line) and partial (dotted line) information starting from $w = 1$. The two distributions have the same mean, and are obtained by setting $\gamma=2.08$ for the partially informed investor and $\gamma=4.03$ for the fully informed one.}
\label{pdf_bis}
\end{figure}

\begin{figure}[h]
{\includegraphics[height = 6cm, width = 12cm]{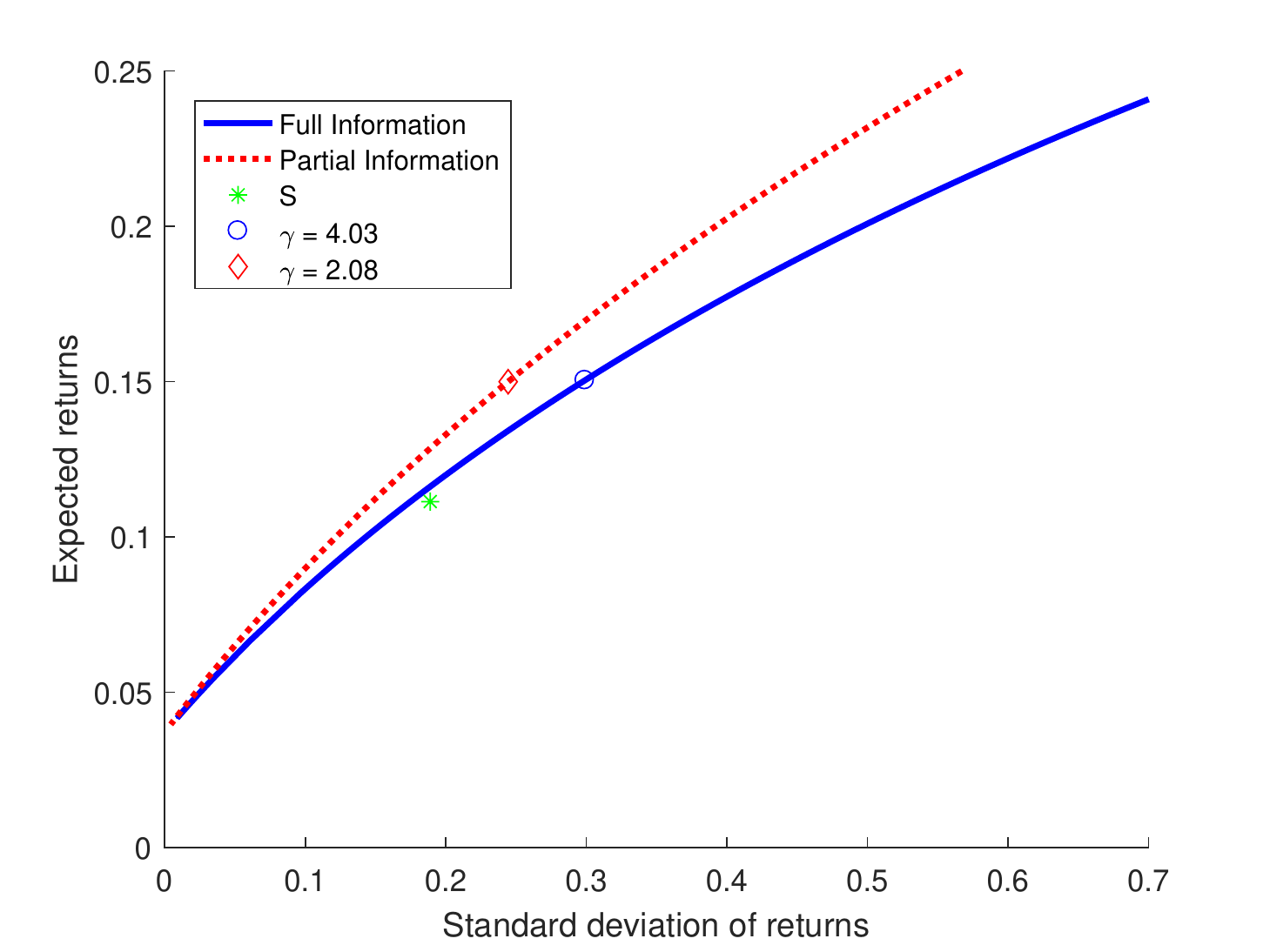}}
\caption{Expected returns of optimal strategies under full (continuous line) or partial (dotted line) information as functions of their standard deviations. The points obtained for $\gamma=2.08$ with partial information and   $\gamma=4.03$ with full information are reported, for reference with Figure \ref{pdf_bis}. The point $S$ represents the risky asset. }
\label{Eff_frontier}
\end{figure}


 The cumulative probability distributions for the optimal final wealth under full and partial information are represented in Figure \ref{cdf}.
 The plot shows that the optimal wealth under full information (continuous line) stochastically
dominates the optimal wealth in partial information (dotted line). However such a dominance is lost if the partially informed investor adds to the initial budget $w$
the reservation price  for Dynamic Information $\Delta w$. In this case, by definition, the investor attains the same expected utility as the fully informed investor and
hence her optimal wealth is sometimes  lower sometimes higher  than the one  obtained by the fully
informed investor.
\begin{figure}[htpb]
{\includegraphics[height = 6cm, width = 12cm]{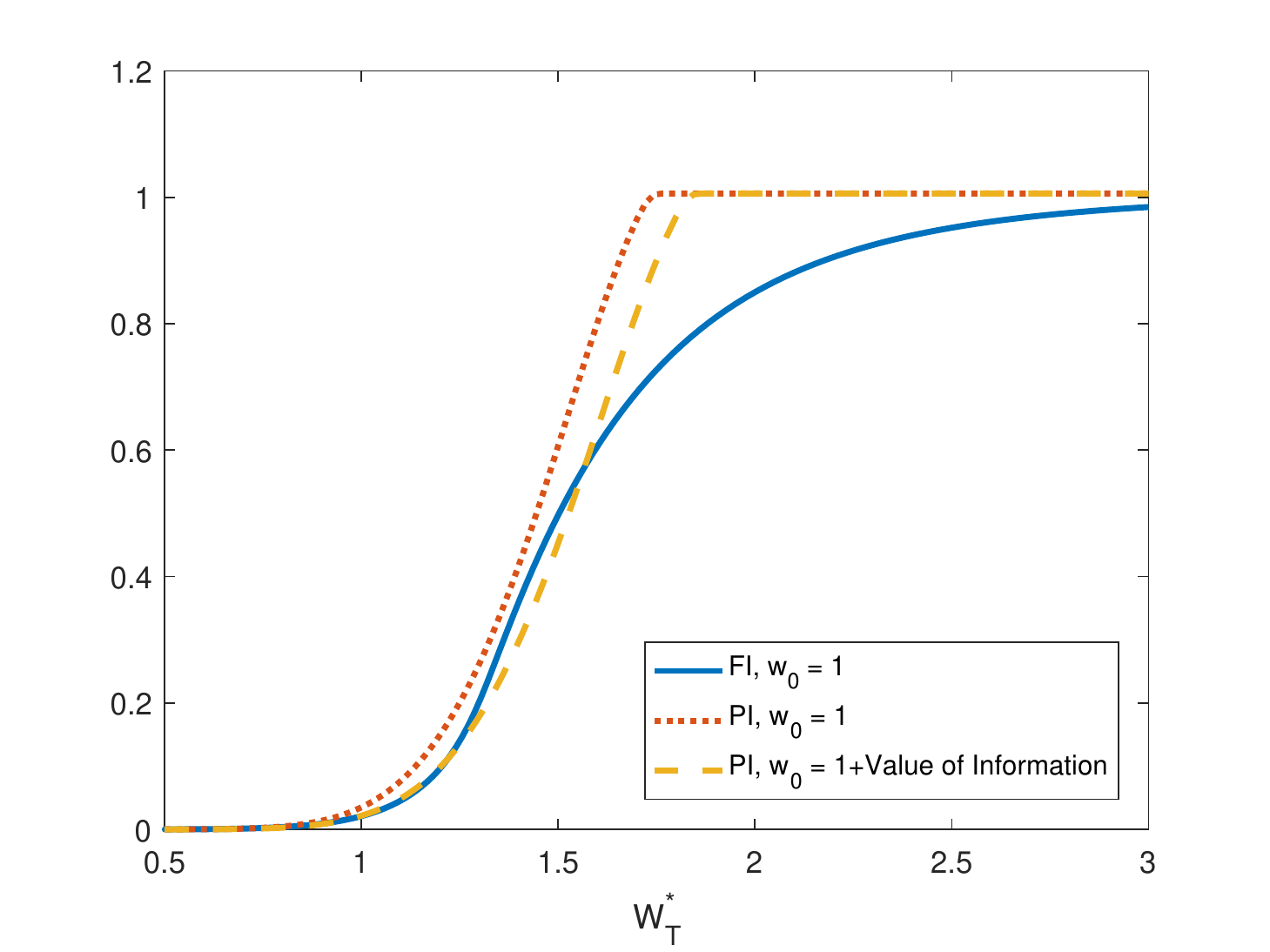}}
\caption{Cumulative probability distribution for the optimal final wealth under full (continuous line) and partial (dotted line) information starting from $w = 1$ and for
the optimal wealth under partial information starting from $w = 1 + \Delta w$ (dashed line), where $\Delta w$ is the reservation price of Dynamic Information .}
\label{cdf}
\end{figure}

 The certainty equivalent of the optimal utility under partial information with respect to the initial conditional variance $R_0$, computed from
\eqref{U_PI}, is represented in Figure \ref{CE_PI_vsR0}.
 The expected utility does not always grow as the precision of the initial estimates increases. In particular, for different values of $\rho$, the
 certainty equivalent is either increasing or it takes on the minimum value within the interval $(0.1,1)$. The intuitive explanation for this fact is that, when the expected
 value of the market price of risk $\pi_0$ is fixed,  a greater uncertainty on its estimate may increase the likelihood of a better or a more favorable outcome, consequently raising
 the expected utility of the optimal wealth.

Figure \ref{ValueInfoR0_onlyx0}  shows the value of the Initial Information ${\mathcal V}^I$ (see Equation \eqref{InitialInfo}) as a function of  $R_0$, for three  values of the
correlation $\rho$. As  expected, the higher the uncertainty on the initial
estimate, the higher ${\mathcal V}^I$. It is perhaps less expected  that  the value  is  greater for $\rho=-0.9$ than for the other two cases.  Why is  the
investor willing to pay a larger share of her initial wealth when the correlation of the changes in the market price of risk with the stock returns is more negative? In our
opinion, this is a combination of two effects: the first effect is related to the precision of the estimate of the market price of risk, the second effect to the
expected return of the optimal strategy. To explain the first effect, we note that, when $\rho=0.9$, the variance of the estimate, $R(t)$, decreases faster to the steady
state
$R_\infty=0.0092$, while for $\rho=-0.9$ and $\rho=0$, it decreases, at a slower rate, towards
$R_\infty=0.0632$ and $R\infty=0.0769$, respectively. Hence  a more accurate information on $X_0$  must be worth less when $\rho=0.9$.
As for the second effect, the certainty equivalent of the optimal strategy under partial information when $R_0=0$ obtained from \eqref{phisTtilde},  is $32.11\%$ of the initial wealth for $\rho=-0.9$, $26.11\%$ for $\rho=0$ and $26.58\%$ for $\rho=0.9$. Therefore, when
$\rho=-0.9$, the investor
is expecting a higher return, and hence she is willing to invest a larger share of her initial wealth to know the exact value of $X_0$.
\begin{figure}[h]
{\includegraphics[height = 6cm, width = 12cm]{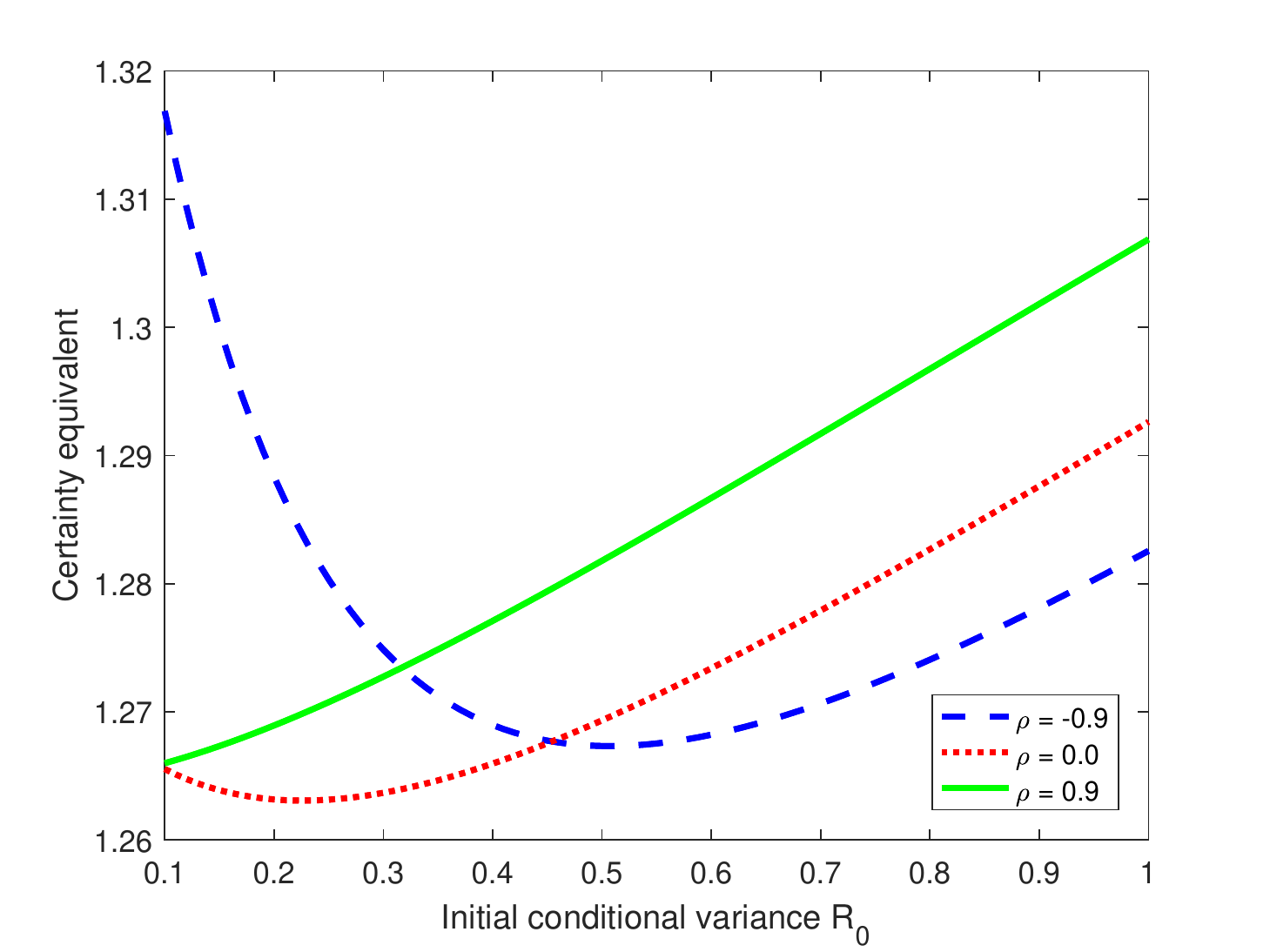}}
\caption{The certainty equivalent under partial information computed from \eqref{U_PI} as a function of the initial variance of the estimate $R_0$.}
\label{CE_PI_vsR0}
\end{figure}

\begin{figure}[h]
{\includegraphics[height = 6cm, width = 12cm]{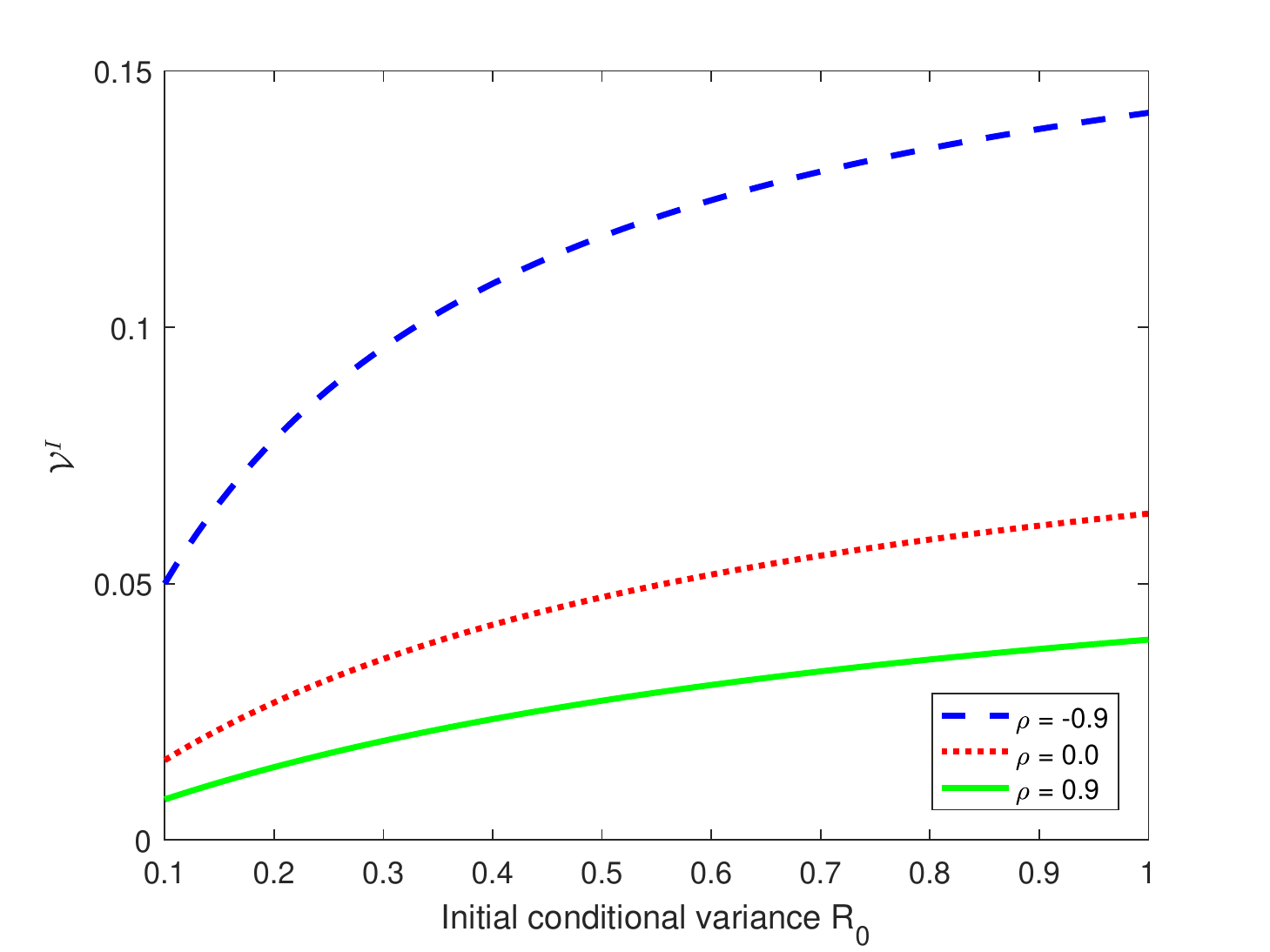}}
\caption{The value of Initial Information \eqref{InitialInfo} as a function of the initial variance of the estimate $R_0$, for different correlation $\rho$.}
\label{ValueInfoR0_onlyx0}
\end{figure}

\begin{figure}[h]
{\includegraphics[height = 6cm, width = 12cm]{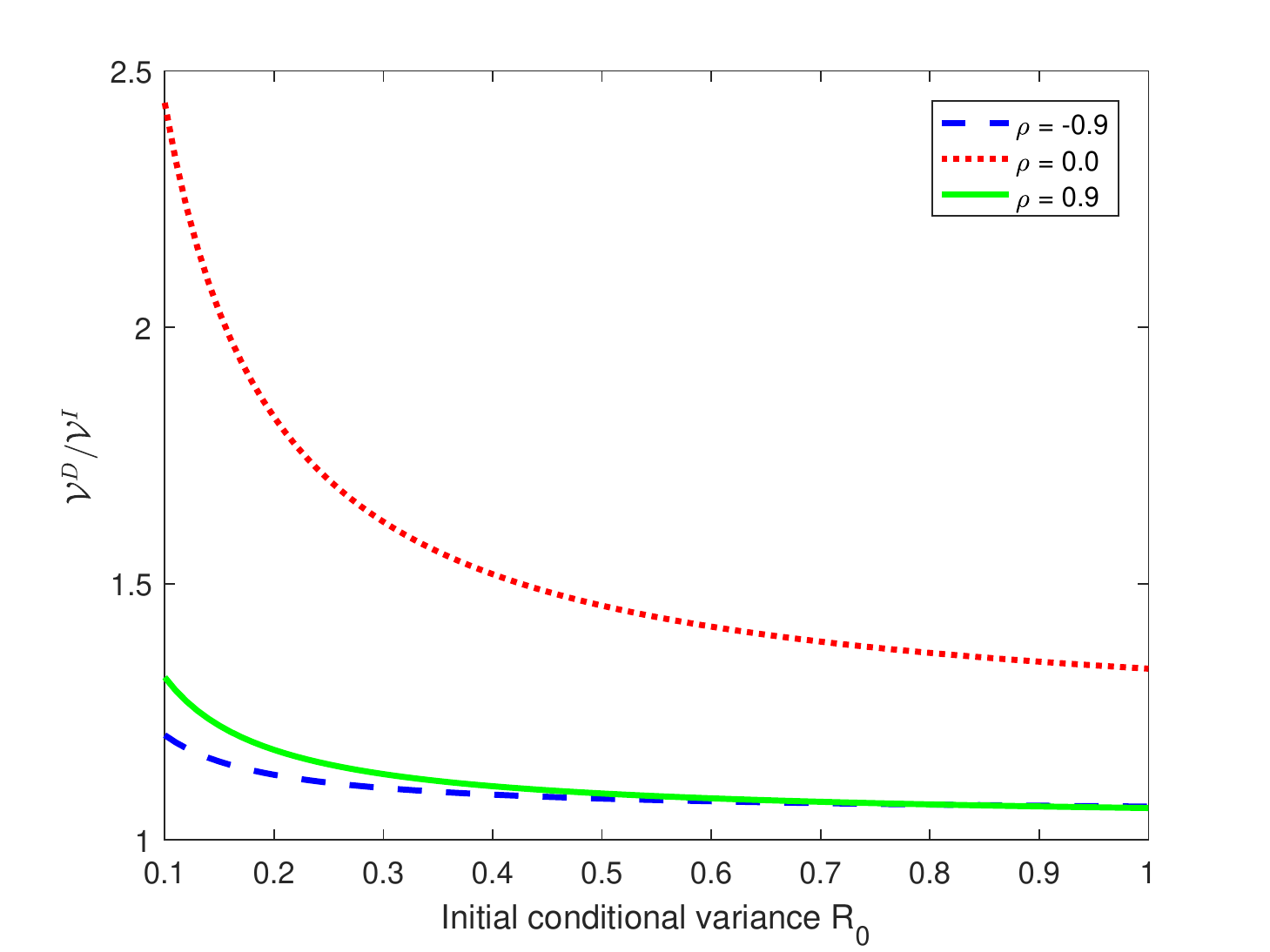}}
\caption{The ratio of the values of information:  Dynamic Information, Equation \eqref{valueInfo}, over Initial Information, Equation \eqref{InitialInfo}, as a function of the initial uncertainty $R_0$.}
\label{RatioValueInfoR0_rho.pdf}
\end{figure}

Figure \ref{RatioValueInfoR0_rho.pdf} presents the ratio of the value of Dynamic Information ${\mathcal V}^D$,  \eqref{valueInfo}, over  the value of Initial Information
${\mathcal V}^I$ \eqref{InitialInfo},  as a function of the initial uncertainty $R_0$, and for different values of $\rho$.
The ratio   is always positive and greater than $1$  because of \eqref{VIVD}. It is decreasing with $R_0$ and  converges to a constant as $R_0$ increases.
When $R_0$ goes to zero,  ${\mathcal V}^I$ also goes to zero while ${\mathcal V}^D$  converges to a positive value, hence  the ratio  grows unbounded.
The ratio is larger for $\rho=0$ and the difference between $\rho=0.9$ and $\rho=-0.9$ is  small. Intuitively, when the correlation is close to $1$ or $-1$, the knowledge of  the starting value for the process $X$ is  sufficient to
estimate with good precision also its next values, and hence the value added by the full knowledge of $X$ is low (for our set of parameters it is around $5\%$
of the value of knowing only $X_0$). Instead, when there is no correlation ($\rho=0$), knowing $X_0$ alone is not sufficient to get a good future estimate for $X$,   and hence the  value added by the dynamic information is more appreciated by the investor.

Figure \ref{Ratio_maturity.pdf}  provides the  ratio ${\mathcal V}^D/{\mathcal V}^I$ as a function of  the investment horizon $T$, for a fixed value $R_0=0.09$. The ratio
increases almost linearly with $T$, but more steeply for $\rho=0$, that is when having access to a dynamic information on the market price of risk  adds a significant
improvement to the investment policy.

\begin{figure}[h]
{\includegraphics[height = 6cm, width = 12cm]{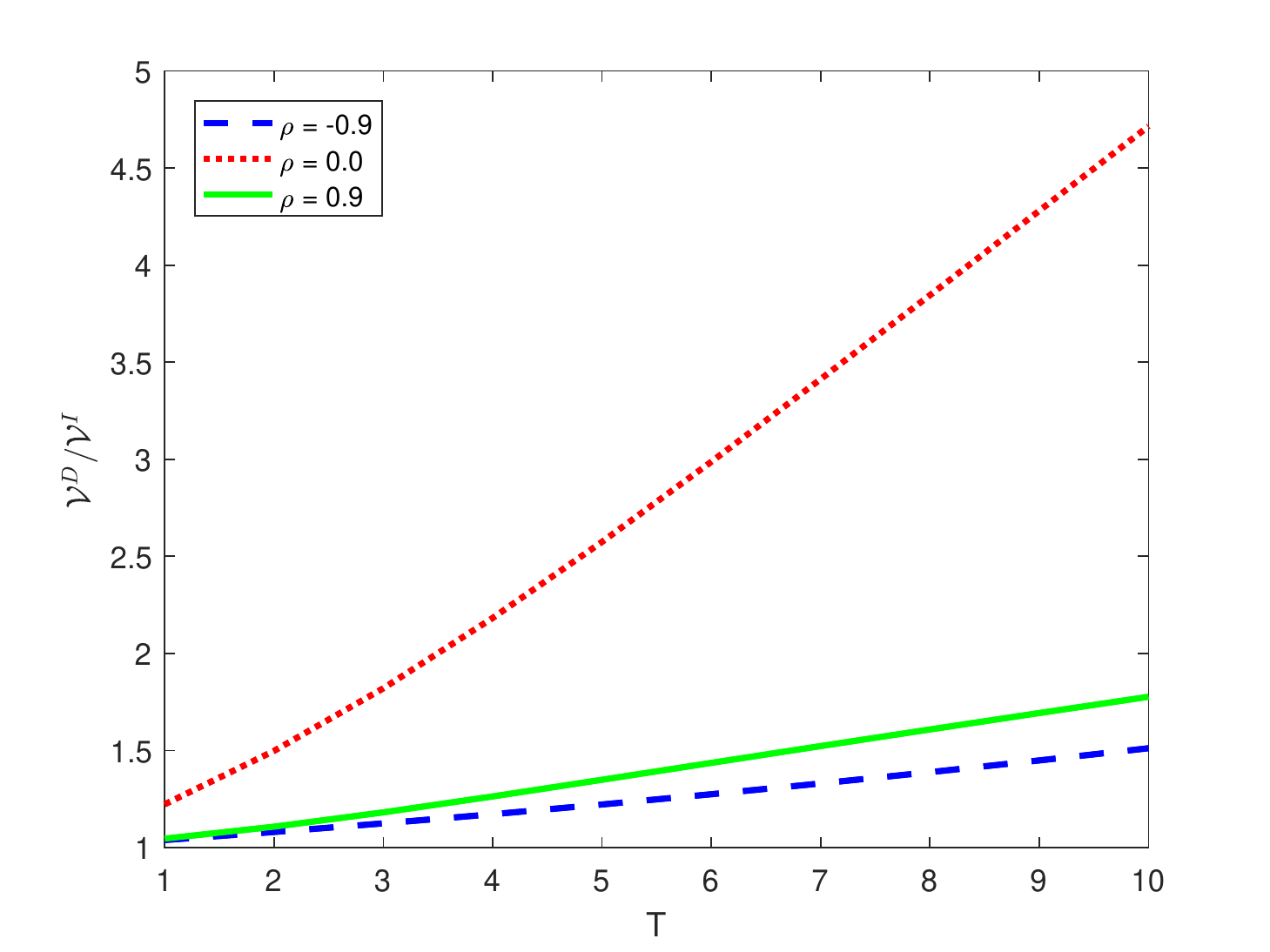}}
\caption{The ratio of the values of information:  dynamic information \eqref{valueInfo} over initial information \eqref{InitialInfo}, as a function of the length of the investment period $T$, for fixed
$\gamma=5$.}
\label{Ratio_maturity.pdf}
\end{figure}

\section{Conclusions}\label{sec:conclusions}

We  studied a portfolio optimization problem for an investor who aims to maximize her expected utility from terminal wealth under two different hypotheses on the information flows when the market price of risk is stochastic and mean-reverting.
We solved the problem via the martingale approach and found an explicit representation for the optimal wealth and its associated utility as function of the current state-price density process and of the market price of risk $X$ in the full information case, or of its best estimate $\pi$ under partial information. We also provided verification theorems for our results.

We introduced the notion of value of information as the maximum percentage of the initial wealth that an investor would be willing to pay to access to more accurate information  on the market price of risk $X$.  In particular we  considered the value of knowing the whole path of $X$ on-the-run and the value of knowing only
its initial value $X_0$.  Using the structure of the solutions of the Riccati equations that characterize the optimal wealth, we determined  closed form representations of such values.
We  provided  applications to illustrate some consequences of our results. The empirical analysis of the distribution of the optimal wealth under full and partial information showed several features that could not be guessed a priori, like for instance the fact that, under our parameter setting,  an investor who cares for the Sharpe ratio of her investment would better allocate her wealth to a partially informed portfolio manager rather than to a fully informed one.
Our measure for the value of information  can be applied to real market data,  for example to determine in which periods of time the
access to a better knowledge on the market price of risk is more valuable. Our approach may also be used to assess the value of an improvement of the initial
prior on the market parameters, and consequently to address issues related to the evaluation of model error.

\begin{center}
{\bf Acknowledgements}
\end{center}
The authors would like to thank the Referees for their useful suggestions. The work on this paper initiated when Katia Colaneri was visiting the Department of Economics,
University of Perugia, as a part of the ACRI Young Investigator Training Program (YITP): the Association
of Italian Banking Foundations and Savings Banks (ACRI) partially supported this work. Katia Colaneri also
received partial financial support from INdAM-GNAMPA under grant UFMBAZ-
2018/000349 and UFMBAZ-2019/000436. The research of Stefano Herzel and Marco Nicolosi was partially funded by the Swedish Research Council grant 2015-01713. Part of this article was written while Katia Colaneri was affiliated with the School of Mathematics of the University of Leeds (UK).

{}

\end{document}